\documentclass[12pt]{amsart}
\usepackage{eurosym}

\usepackage{amsmath}
\usepackage{amsfonts,xcolor}
\usepackage{amssymb}
\usepackage[all]{xy}
\usepackage{amssymb}
\usepackage{subfigure}
\usepackage{graphicx}

\numberwithin{equation}{section}
\makeatletter
\newtheorem*{rep@theorem}{\rep@title}
\newcommand{\newreptheorem}[2]{
\newenvironment{rep#1}[1]{
 \def\rep@title{#2 \ref{##1}}
 \begin{rep@theorem}}
 {\end{rep@theorem}}}
\makeatother

\newcommand{\ceil}[1]{\ensuremath{\lceil #1 \rceil}}
\newcommand{\bN}{{\mathbb{N}}}
\newcommand{\bR}{{\mathbb{R}}}
\newcommand{\Rdst}{{\mathbb{R}^d}}
\newcommand{\Rst}{{\mathbb{R}}}

\newcommand{\Nst}{{\mathbb{N}}}
\newcommand{\Rtdst}{{\mathbb{R}^{2d}}}

\newcommand{\norm}[1]{\lVert#1\rVert}
\newcommand{\bC}{{\mathbb{C}}}

\newcommand{\Bignorm}[1]{\Bigl\|#1\Bigr\|}
\newcommand{\bignorm}[1]{\bigl\lVert#1\bigr\rVert}

\newcommand{\abs}[1]{\ensuremath{\left| #1 \right| }}

\newcommand{\comms}[1]{\mbox{}}

\newcommand{\gabspace}{\mathcal{V}}

\newcommand{\wigfg}{W(f,g)}

\newreptheorem{theorem}{Theorem}
\newtheorem{lemma}{Lemma}[section]
\newtheorem{theorem}[lemma]{Theorem}
\newtheorem{coro}[lemma]{Corollary}
\newtheorem{prop}[lemma]{Proposition}

\newtheorem{rem}[lemma]{Remark}
\newtheorem{remark}[lemma]{Remark}
\newtheorem{definition}[lemma]{Definition}

\def\rd{\bR^d}

\def\rdd{{\bR^{2d}}}

\subjclass{}
\thanks{
L.\ D.\ A.\ was supported by the Austrian Science Fund (FWF): START-project FLAME (”Frames and Linear Operators for
Acoustical Modeling and Parameter Estimation”, Y 551-N13),
and by FWF P. 31225-N32.
K.\ G.\ was
supported in part by the  project P31887-N32 
of the Austrian Science Fund (FWF), 
J. L. R. gratefully acknowledges support from the Austrian Science Fund (FWF):P 29462-N35, and from the WWTF grant INSIGHT (MA16-053).}

\begin{document}
\title{Harmonic analysis in phase space and finite Weyl-Heisenberg ensembles}
\author{Lu\'{\i}s Daniel Abreu}
\address{Acoustics Research Institute\\
Austrian Academy of Sciences\\
Wohllebengasse 12-14, Vienna, 1040, Austria.}
\email{labreu@kfs.oeaw.ac.at}

\author{Karlheinz Gr\"{o}chenig}
\address{Faculty of Mathematics, University of Vienna, Oskar-Morgenstern-Platz 1, 1090 Vienna, Austria.}
\email{karlheinz.groechenig@univie.ac.at}

\author{Jos\'{e} Luis Romero}
\address{Faculty of Mathematics, University of Vienna, Oskar-Morgenstern-Platz 1, 1090 Vienna, Austria\\and\\Acoustics Research Institute, Austrian Academy of Science, Wohl\-leben\-gasse 12-14, 1040 Vienna, Austria.}
\email{jose.luis.romero@univie.ac.at, jlromero@kfs.oeaw.ac.at}
 
\date{}

\begin{abstract}
Weyl-Heisenberg ensembles are translation-invariant determinantal point processes 
on $\mathbb{R}^{2d}$ associated with
the Schr\"odinger representation of the Heisenberg group, and include as
examples the Ginibre ensemble and the polyanalytic ensembles, which model
the higher Landau levels in physics. We introduce finite versions of the
Weyl-Heisenberg ensembles and show that they behave analogously to the
finite Ginibre ensembles. More specifically, guided by the observation that the Ginibre ensemble
with $N$ points is asymptotically close to the restriction of the infinite Ginibre ensemble
to the disk of area $N$, we define finite WH ensembles as adequate finite approximations
of the restriction of infinite WH ensembles to a given domain $\Omega$. We provide a precise rate for the convergence of the corresponding one-point intensities to the indicator function of $\Omega$, as $\Omega$ is dilated and the process is rescaled proportionally (thermodynamic regime). The construction and analysis rely neither on explicit formulas nor on the asymptotics for orthogonal polynomials, but rather on phase-space methods.

Second, we apply our construction to study the pure finite Ginibre-type polyanalytic ensembles,
which model finite particle systems in a single Landau level, and are
defined in terms of complex Hermite polynomials. On a technical level, we
show that finite WH ensembles provide an approximate model for finite
polyanalytic Ginibre ensembles, and we quantify the corresponding deviation.
By means of this asymptotic description, we derive estimates for
the rate of convergence of the one-point intensity of polyanalytic
Ginibre ensembles in the thermodynamic limit. 
\end{abstract}

\maketitle

\section{Introduction}

\subsection{Weyl-Heisenberg ensembles}

We study the class of determinantal point processes on $\mathbb{R}^{2d}$
whose correlation kernel is given as
\begin{equation}  \label{eq:l1}
{K^g}((x,\xi),(x^{\prime },\xi ^{\prime }))=\int_{\mathbb{R}^{d}}e^{2\pi
i(\xi^{\prime } -\xi)t}g(t-x^\prime) \overline{g(t-x )}dt
\end{equation}
for some non-zero (normalized) function $g\in L^2(\bR^d )$ and $(x,\xi),
(x^{\prime },\xi^{\prime }) \in \mathbb{R}^{2d}$. These determinantal point
processes are called Weyl-Heisenberg ensembles (WH ensembles) and have been
introduced recently in \cite{APRT}. They form a large class of
translation-invariant hyperuniform point processes \cite{TorStil, PHRE2009,
ghosh1}.

The prototype of a Weyl-Heisenberg ensemble is the complex Ginibre ensemble.
Choosing $g$ in~\eqref{eq:l1} to be the Gaussian $g(t)=2^{1/4}e^{-\pi t^{2}}$
and writing $z=x+i\xi, z^{\prime }= x^{\prime }+i\xi^{\prime }$, the
resulting kernel is then

\begin{equation}  \label{eq_intro_kg}
{K^g}(z,z^{\prime})=e^{i\pi (x^{\prime }\xi^{\prime }-x\xi )} e^{-\frac{\pi}{2}%
(\left| z \right|  ^{2}+\left| z^{\prime }\right|  ^{2})} e^{\pi \overline{z}
z^{\prime }},
\qquad 
z=x+i\xi, \, z^{\prime }= x^{\prime }+i\xi^{\prime }.
\end{equation}
Modulo conjugation with a phase factor, this is essentially the kernel of
the \emph{infinite Ginibre ensemble} $K_\infty (z,z^{\prime})= e^{-\frac{\pi}{2}
(\left\vert z\right\vert ^{2}+\left\vert z^{\prime }\right\vert ^{2})}e^{\pi
z\overline{z^{\prime }}}$. Another important class of examples arises by
choosing $g$ to be a Hermite function. In this case one obtains a pure
polyanalytic Ginibre ensemble ~\cite{SHIRAI, APRT}, which models the
electron density in a single (pure) higher Landau level (see Section \ref%
{sec_pure_pol} for some background).

The Ginibre ensemble with kernel $K_\infty$ arises as limit of corresponding
processes with $N$ points, whose kernels
\begin{equation}
K_{N}(z,z^{\prime})= e^{-\frac{\pi}{2} (\left\vert z\right\vert ^{2}+\left\vert
z^{\prime }\right\vert ^{2})}\sum_{j=0}^{N-1}\frac{\left( \pi z\overline{%
z^{\prime }}\right) ^{j} }{j!},  \label{finite}
\end{equation}
are obtained simply by truncating the expansion of the exponential $e^{\pi z%
\overline{z^{\prime }}}$. It is not obvious how to obtain the analogous
finite-dimensional process for a general Weyl-Heisenberg ensemble %
\eqref{eq:l1}, because for most choices of $g\in L^{2}({\mathbb{R}}^d)$
there is no treatable explicit formula available for ${K^g}$. We present a
canonical construction of finite Weyl-Heisenberg ensembles and show that
they enjoy properties similar to the finite Ginibre ensemble. The
construction and analysis is based on spectral theory of Toeplitz-like
operators and harmonic analysis of phase space.

The abstract construction is instrumental to study the asymptotic properties
of a particularly important class of finite-dimensional determinantal point
processes, namely the finite pure polyanalytic Ginibre ensembles, which
model the electron density in higher Landau levels. This is an example where
the Plancherel-Rotach asymptotics of the basis functions are not available.
Moreover, the relevant polynomials do not satisfy the classical three-term
recurrence relations which are used in Riemann-Hilbert type methods \cite%
{DeiftRH, DKMVZ}. We develop a new approach based on spectral methods and harmonic analysis in phase space and show that the finite WH ensembles
associated with a Hermite function are asymptotically close to finite polyanalytic ensembles.
Thus, our analysis of the finite polyanalytic ensembles has two steps: (i)
the abstract construction of finite WH ensembles and their thermodynamic
limits; (ii) the comparison of the finite WH ensembles associated with
Hermite functions and the finite pure polyanalytic ensembles.

\subsection{Planar Hermite ensembles}

The complex Hermite polynomials are given by
\begin{equation}
H_{j,r}(z,\overline{z})=\left\{
\begin{tabular}{l}
${\sqrt{\frac{r!}{j!}}\pi ^{\frac{j-r}{2}}z^{j-r}L_{r}^{j-r}\left( \pi
\left\vert z\right\vert ^{2}\right), \qquad j>r \geq 0,}$ \\
${\left( -1\right) ^{r-j}\sqrt{\frac{j!}{r!}}\pi ^{\frac{r-j}{2}}\overline{ z%
}^{r-j}L_{j}^{r-j}\left( \pi \left\vert z\right\vert ^{2}\right), \qquad 0
\leq j\leq r}$,%
\end{tabular}
\right.  \label{ComplexHermite}
\end{equation}
where $L_{r}^\alpha $ denotes the Laguerre polynomial
\begin{align}  \label{eq_lag}
&L_{j}^{\alpha}(x)=\sum\limits_{i=0}^{j}(-1)^{i}\binom{j+\alpha }{j-i}\frac{%
x^{i}}{i!}, \qquad x\in {\mathbb{R}}, \qquad j \geq 0, j+\alpha \geq 0.
\end{align}
Complex Hermite polynomials satisfy the doubly-indexed orthogonality
relation
\begin{equation*}
\int_{\mathbb{C}}H_{j,r}(z,\overline{z})\overline{H_{j\prime ,r\prime }(z,%
\overline{z})}e^{-\pi \left\vert z\right\vert ^{2}}dz={\delta }_{jj\prime }{%
\delta }_{rr\prime },
\end{equation*}
and form an orthonormal  basis of  $L^{2}\left( \mathbb{C}, e^{-\pi
\left\vert z\right\vert ^{2}}\right)$ \cite{AbrGr2012} \footnote{%
Perelomov \cite{Perelomovbook} mentions that \eqref{ComplexHermite} has been
used by Feynman and Schwinger as the explicit expression for the matrix
elements of the displacement operator in Bargmann-Fock space.}.

The complex Hermite polynomials form a complete set of eigenfunctions of the  Landau operator 
\begin{equation}
L_{z}:=-\partial _{z}\partial _{\overline{z}}+\pi \overline{z}\partial _{%
\overline{z}}   \label{2.1.3}
\end{equation}%
acting on the Hilbert space $L^{2}( \mathbb{C}, e^{-\pi
\left\vert z\right\vert ^{2}}) $. The Landau operator is the
Schr\"odinger operator that models the behavior of an electron in $\mathbb{R}^2$ in a
constant magnetic field perpendicular to the $\mathbb{C}$-plane.
The spectrum of $L_{z}$, i.e., the set of possible energy levels, is given by $\sigma (L_{z})=\{r\pi :r=0,1,2,\ldots \}
$ and the  eigenspace associated with the eigenvalue $r\pi $ is called the \emph{%
 Landau level of order }$r$. For the minimal energy $r=0$, i.e., 
the ground state, the eigenspace is the classical Fock space, for
$r>0$,   the eigenspaces are spanned by the orthonormal basis $\{
H_{j,r}: j\in \bN \}$.  The Landau levels are key for the
mathematical formulation of the integer quantum Hall effect
discovered by von Klitzing \cite{Nobel}.

We will consider a variety of ensembles associated with the
complex Hermite polynomials. 

\begin{definition}
Let $J \subseteq\bN_{0}\times \bN_{0}$. The planar Hermite ensemble based on
$J$ is the determinantal point process\ with the correlation kernel
\begin{equation}
K(z,z^{\prime})= e^{-\frac{\pi}{2} ({\left\vert z\right\vert ^{2}+\left\vert
z^{\prime }\right\vert ^{2}})}\sum_{j,r\in J}H_{j,r} \left(z,\overline{z}
\right) \overline{ H_{j,r} \left( z^{\prime },\overline{z^{\prime }} \right)}%
.  \label{complexhermiteensembles}
\end{equation}
\end{definition}

Complex Hermite polynomials are an example of \emph{polyanalytic functions}
- that is, polynomials in $\overline{z}$ with analytic coefficients (see
Section \ref{app_poly}). While most classes of orthogonal polynomials
satisfy a three-term recurrence relation - which puts them in the scope of
Riemann-Hilbert type techniques \cite{DeiftRH,DKMVZ} - the complex Hermite
polynomials satisfy instead a system of doubly-indexed recurrence relations
\cite{Ismail,Ghanmi}.

Several important determinantal point processes arise as special cases of (%
\ref{complexhermiteensembles}). First, since $H_{j,0}(z,\overline{z}%
)=(\pi^{j}/j!)^{\frac{1}{2}}z^{j}$, the set $J=\{0,\ldots,N-1\}\times \{0\}$
in \eqref{complexhermiteensembles} leads to the kernel of the Ginibre
ensemble \eqref{finite}. A second important example arises for $J :=
\{(j,r): 0 \leq j \leq n-1, r=m-n+j\}$ with $n,m \in \mathbb{N}$. The
corresponding one-point intensity is a radial version of the marginal
probability density function of the unordered eigenvalues of a complex
Gaussian Wishart matrix after the change of variables $t \to \pi \left| z
\right|  ^2$, see, e.g. \cite[Theorem 2.17]{verdu}. Thirdly, choosing $%
J=\{0,\ldots,N-1\}\times \{0,\ldots,q-1\}$ one obtains the polyanalytic
Ginibre ensemble introduced by Haimi and Hedenmalm \cite{HendHaimi}. The
polyanalytic Ginibre ensemble gives the probability distribution of a system
composed by several Landau levels. The case of more general interaction potentials has been investigated in \cite{HendHaimi,HaiHen2}, by considering
polyanalytic Ginibre ensembles with general weights. These investigations parallel the ones of weighted Ginibre ensembles \cite{AHM0,AMH1,AMH2}.

We are particularly interested in
finite versions of the infinite pure polyanalytic ensembles defined by
Shirai \cite{SHIRAI}.
The infinite ensembles  are defined by  the reproducing kernels of an  eigenspace of the
Landau operator~\eqref{2.1.3} which is  given by   
\begin{equation*}
K_{r}(z,z^{\prime })=L_{r}^{0}(\pi \left\vert z-z^{\prime }\right\vert
^{2})e^{\pi z\overline{w}-\tfrac{\pi }{2}(\left\vert z\right\vert
^{2}+\left\vert z^{\prime }\right\vert ^{2})}=e^{-\frac{\pi }{2}(\left\vert
z\right\vert ^{2}+\left\vert z^{\prime }\right\vert ^{2})}\sum_{j=0}^{\infty
}H_{j,r}\left( z,\overline{z}\right) \overline{H_{j,r}\left( z^{\prime },%
\overline{z^{\prime }}\right) }\text{.}
\end{equation*}
Here the second identity follows from the fact that  $\left\{
  H_{j,r}\left( z,\overline{z}\right) \right\} _{j\in \mathbb{N}% 
}$ spans  the $rth$  eigenspace of the Landau operator.
The corresponding finite pure polyanalytic ensembles can now be defined as
planar Hermite ensembles with $J=\{0,\ldots,N-1\}\times \{r\}$. In
analogy to \eqref{finite},
the finite \emph{$(r,N)$-pure polyanalytic ensemble} is the determinantal
point process with correlation kernel
\begin{equation}  \label{purekernel}
K_{r,N}(z,z^{\prime})= e^{-\frac{\pi}{2} (\left\vert z\right\vert ^{2}+\left\vert
z^{\prime }\right\vert ^{2})} \sum_{j=0}^{N-1} H_{j,r} \left(z,\overline{z}
\right) \overline{H_{j,r}\left(z^{\prime },\overline{z^{\prime }}\right)}.
\end{equation}
While pure polyanalytic ensembles describe individual Landau levels, their finite counterparts model a finite number of particles confined to a certain disk (for example, as the result of a radial potential). In this article, we prove the following theorem, which supports this interpretation, and provides a rate of convergence for the one-point intensity related to each Landau level.

\begin{theorem}
\label{asy_pure} Let $\rho_{r,N}(z)=K_{r,N}(z,z)$ be the one-point intensity
of the finite $(r,N)$-pure polyanalytic Ginibre ensemble. Then, for each $r>0$,
\begin{equation}
\rho_{r,N} \Big(\sqrt{\tfrac{N}{\pi}} \, \cdot \Big)\longrightarrow 1_{%
\mathbb{D}},
\end{equation}
in $L^{1}({\mathbb{R}^{2}})$, as $N\longrightarrow +\infty $. Moreover,
\begin{equation}
\Bignorm{ \rho_{r,N}-1_{\mathbb{D}_{\sqrt{N/\pi}}} } _{1}\leq C_{r} \sqrt{N}.
\label{est1.1}
\end{equation}
\end{theorem}

The convergence rate in Theorem \ref{asy_pure} is independent of the
energy level $r$ of the Landau operator. It is known to be sharp for
the first Landau level $r=0$, and we believe that \eqref{est1.1} is
also sharp for all Landau levels $r\in \bN $. 
\footnote{The first Landau level is also called \emph{ground level}
because it corresponds to the lowest energy.}
\footnote{See also \cite[Proposition 14]{Meckes}
and \cite{Equidistribution}, where it is pointed out that
the sharp rate for the ground level also follows from pointwise estimates for Bergman kernels \cite{Tian}.}

In statistical terms, \eqref{est1.1} means that the number of points of the $(r,N)$-pure polyanalytic Ginibre ensemble that belong to a certain domain
$A \subseteq \mathbb{C}$, $n_{r,N}(A)$, satisfies
\begin{align}
\label{eq_ff}
\mathbb{E}\{ n_{r,N}(A) \} = \abs{\mathbb{D}_{\sqrt{N/\pi}} \cap A} + \mathcal{O}\big(\sqrt{N}\big).
\end{align}
Theorem \ref{asy_pure} supports and validates the interpretation of finite pure polyanalytic ensembles as models for $N$ particles confined to a disk by giving asymptotics for the first order statistics \eqref{eq_ff} that indeed show concentration on the disk area $N$,
up to an error comparable to the perimeter of that disk. In addition, \eqref{eq_ff} implies that, after proper rescaling, the particles are, in expectation, asymptotically equidistributed on the disk. This statistical description is consistent with the notion of a \emph{filling factor} of each Landau level - that is, a certain limit to the number of particles that each level can accommodate. The incremental saturation of each individual Landau level, corresponding to incremental energy levels, is part of the mathematical
description of the integer quantum Hall effect discovered by von Klitzing \cite{Nobel}. 
(The integer quantum Hall effect is not to be confused with the fractional quantum Hall effect, whose mathematical formulation
is related to the Laughlin's wave function \cite{Laughlin} and the
so-called beta-ensembles \cite{BEY, BBNY}.)

As a first step towards a description of finite pure polyanalytic ensembles,
we introduce a general construction of finite versions of Weyl-Heisenberg
ensembles that may be of independent interest.

\subsection{Finite Weyl-Heisenberg ensembles}

\label{sec_intro_cons} The construction of finite WH ensembles relies on
methods from harmonic analysis on phase space \cite{folland89, MR2226126},
and on the spectral analysis of phase-space Toeplitz operators. Write $%
z=(x,\xi) \in {\bR^{2d}}, z^{\prime }= (x^{\prime },\xi ^{\prime }) \in {\bR%
^{2d}} $ for a point in phase space and
\begin{equation}  \label{eq:l2}
\pi (z)f(t):=e^{2\pi i\xi t}f(t-x)
\end{equation}
for the phase-space shift by $z$. Then the kernel in \eqref{eq:l1} is given
by
\begin{equation}  \label{eq:l3}
{K^g}(z,z^{\prime }) = \langle \pi (z^{\prime })g, \pi (z)g\rangle.
\end{equation}
Let us now describe the construction of the finite point processes
associated with the kernel ${K^g}$. For normalized $g\in L^2(\bR^d )$, $%
\|g\|_2 = 1$, the integral operator with kernel ${K^g}$, i.e., $F
\mapsto \int_{\mathbb{R}^{2d}} {K^g}(z,z^{\prime })
F(z^{\prime }) dz^{\prime }$,  is an orthogonal
projection (see for example ~\cite[Chapter 1]{folland89}, \cite[Chapter 9]%
{Charly}). Consequently, the  range of this projection  is a \emph{reproducing kernel Hilbert
  space} $\mathcal{V}_g
\subseteq L^2({\bR^{2d}} )$  with the explicit description
\begin{equation*}
\mathcal{V}_g =\big\{F \in L^2({\bR^{2d}} ): F(z) = \langle f, \pi (z)
g\rangle, \mbox{ for } f \in L^2(\bR^d)\,\big\} \subseteq {L^{2}({\mathbb{R}%
^{2d}})}. 
\end{equation*}
Thus every $F\in \mathcal{V} _g$ is a phase-space representation of a
function $f$ defined on the configuration space $\mathbb{R}^d$.

\emph{Step 1: Concentration as a smooth restriction}. Let $\mathcal{X}^g$ be
a WH ensemble (with correlation kernel ${K^g}$) and let $\Omega \subseteq %
\Rtdst$ be a measurable set. The restriction of $\mathcal{X}^g$ to $\Omega$
is a determinantal point process (DPP) $\mathcal{X}^g_{|\Omega}$ with
correlation kernel
\begin{align}  \label{eq_intro_ker_o}
{K^g}_{|\Omega}(z,z^{\prime }) = 1_\Omega(z) {K^g}(z,z^{\prime })
1_\Omega(z^{\prime }).
\end{align}
An expansion of the kernel ${K^g}_{|\Omega}$ can be obtained as follows. We
consider the \emph{Toeplitz operator} on $\mathcal{V}_g$ defined by
\begin{equation}
  \label{Toe0}
M_{\Omega }^{g}F(z) =\int_{\Omega }F(z^{\prime \prime
}){K^g}(z,z^{\prime \prime 
})\,dz^{\prime \prime} \, .
\end{equation}
Since $F(z^{\prime \prime} ) = \int _{\rdd } F(z^\prime )
K^g(z^{\prime \prime},z^\prime) \, dz^\prime$ for $F \in \mathcal{V}_g$, $M_{\Omega
}^{g} $ can be expressed as an integral operator
\begin{align}
 M_{\Omega }^{g}F(z) &=\int_{\mathbb{R}^{2d}}F(z^{\prime \prime}) \, 1_\Omega (z^{\prime \prime
 }){K^g}(z,z^{\prime \prime })\,dz^{\prime \prime } \\
&= \int_{\Rtdst} F(z^{\prime }) \left[ \int_{\Rtdst} {K^g} (z,z^{\prime
\prime }) 1_\Omega (z^{\prime \prime }) {K^g} (z^{\prime \prime },z^{\prime
}) \, dz^{\prime \prime }\right] \, dz^{\prime }.  \label{Toe}
\end{align}
 By definition \eqref{Toe0}, $M_{\Omega }^{g}$ acts on a function $F
 \in \mathcal{V}_g$ by multiplication by $1_\Omega$, followed by
 projection onto  $\mathcal{V}_g$.
On the other hand, if $F \in \mathcal{V}_g^\perp$, then the expression
in \eqref{Toe} vanishes. Thus, the formula in \eqref{Toe} defines the
extension of $M_{\Omega }^{g}$ to $L^2(\mathbb{R}^{2d})$ that is $0$
on $L^2(\mathbb{R}^{2d}) \ominus \mathcal{V}_g$. For $\Omega \subseteq
{\bR^{2d}}$ of finite measure, $M^g_\Omega $ is a compact positive
(self-adjoint) operator on $L^2({\bR^{2d}} 
)$; see for example ~\cite{cogr03,ro12}. By the spectral theorem, $M^g_\Omega$ is
diagonalized by an orthonormal set $\{p_{g,j}^{\Omega }: j \in \bN \}
\subseteq \mathcal{V}_g$ of
eigenfunctions, with corresponding eigenvalues $\lambda _j =
\lambda_{j}^{\Omega }$ (ordered non-increasingly):
\begin{equation}  \label{eq_eigenexp}
M_{\Omega}^{g} = \sum_{j \geq 1} \lambda_{j}^{\Omega } \, p_{g,j}^{\Omega }
\otimes p_{g,j}^{\Omega }.
\end{equation}
The key property is that the eigenfunctions $p_{g,j}^{\Omega }$ are \emph{%
doubly-orthogonal}: since 
\begin{align*}
\left<M_{\Omega}^{g} F,F\right>  = \int_\Omega
\left| F(z) \right|  ^2 dz, \qquad F \in L^2(\mathbb{R}^{2d}),
\end{align*}
it follows that
\begin{align*}
\left<p_{g,j}^{\Omega },p_{g,j^{\prime }}^{\Omega }\right> _{L^2(\Omega)}
=\left<M_{\Omega}^{g} p_{g,j}^{\Omega},p_{g,j^{\prime }}^{\Omega }\right>
_{L^2(\Rtdst)} = \lambda_{j}^{\Omega } \delta_{j,j^{\prime }},
\end{align*}
and consequently the restricted kernel has the orthogonal expansion
\begin{align}  \label{eq_intro_ker_o2}
{K^g}_{|\Omega}(z,z^{\prime }) = \sum_{j \geq 1} \left(
p_{g,j}^{\Omega}(z)1_\Omega(z) \right) \cdot \left( \overline{%
p_{g,j}^{\Omega}(z^{\prime })}1_\Omega(z^{\prime }) \right);
\end{align}
see Section \ref{sec_duality} for details. Note that in %
\eqref{eq_intro_ker_o2}, the functions $p_{g,j}^{\Omega}(z)1_\Omega(z)$ are
not normalized. In fact,
\begin{align}  \label{eq:o11}
\int_\Omega \left| p_{g,j}^{\Omega}(z) \right|  ^2 dz = \lambda_{j}^{\Omega
}.
\end{align}
Thus, while in \eqref{eq_intro_ker_o2} the basis functions are restricted to
the domain $\Omega$, the expansion of the Toeplitz operator \eqref{eq_eigenexp} involves the non-truncated functions $p_{g,j}^{\Omega}(z)$ weighted by the measure of their concentration on $\Omega$ \eqref{eq:o11}. We call the
DPP with correlation kernel corresponding to \eqref{Toe} the \emph{concentration} of the full WH ensemble to $\Omega$ and denote it by $\mathcal{X}^{g,\mathrm{con}}_{\Omega}$. This process is thus a smoother
variant of the restricted process $\mathcal{X}^{g}_{|\Omega}$,
because it involves the (smooth) functions $p_{g,j}^{\Omega}(z)$ instead of their truncations $p_{g,j}^{\Omega}(z) 1_\Omega(z)$, which may have discontinuities along $\partial \Omega$. The
construction of DPPs from the spectrum of self-adjoint operators has been
suggested in \cite{BO1,BO2} as an analogue of the construction of DPPs from
the spectral measure of a group. In a related work \cite{Neretin}, a
combination of methods from operator theory and representation theory has been used to show that a
DPP is the spectral measure for an explicit commutative group of
Gaussian operators in the fermionic Fock space.

\emph{Step 2: Spectral truncation}. Since $\left<M_{\Omega}^{g} F,F\right>  =
\int_\Omega \left| F \right|  ^2$, by the min-max principle,
\begin{align}  \label{eq_minimax}
\lambda^\Omega_j = \max \left\{\int_\Omega \left| F(z) \right|  ^2 dz: %
\norm{F}_2=1, F \in \mathcal{V}_g, F \perp p_{g,1}^{\Omega }, \ldots,
p_{g,j-1}^{\Omega} \right\}.
\end{align}
Thus, the eigenvalues $\lambda_{j}^{\Omega }$
describe the best possible simultaneous phase-space concentration of
waveforms within $\Omega$. In particular, \eqref{eq_minimax} implies
that 
\begin{align*}
0 \leq \lambda^\Omega_j \leq 1, \qquad j \geq 1.
\end{align*}

It is well-known that there are $\approx \left| \Omega \right|$ 
eigenvalues $\lambda_{j}^{\Omega }$ that are close to $1$. As a
precise statement
we cite the following \emph{Weyl-type law}: for any $\delta \in (0,1)$,
\begin{equation}  \label{eq:l4}
\left\vert \#\{j:\lambda _{j}^{\Omega }>1-\delta \}-\left\vert \Omega
\right\vert \right\vert \leq \max \left\{\frac{1}{\delta}, \frac{1}{1-\delta} \right\} C_{g} \left| \partial \Omega \right|
_{2d-1},
\end{equation}
where $\left| \partial \Omega \right|  _{2d-1}$ is the perimeter of $\Omega$
(the surface measure of its boundary), and $C_{g}$ is a constant
depending explicitly on $g$. See for instance \cite[Proposition
~3.4]{AGR} or~\cite{DFN02}.  The dependence of the constant $C_g$ on
$g$  is made explicit below  in \eqref{eq_mst}. 

\begin{figure}[tbp]
\centering
\subfigure{
\includegraphics[scale=1]{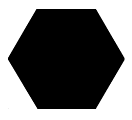}
} \subfigure
{\ \includegraphics[scale=0.3]{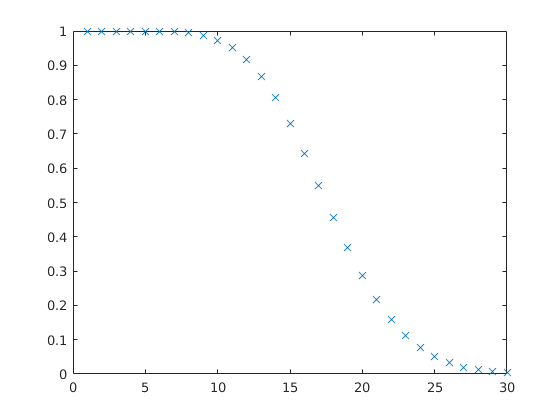} }
\caption{A plot of the eigenvalues of the Toeplitz operator $M^g_\Omega$,
with $g$ a Gaussian window and $\Omega$ of area $\approx$ 18.}
\label{fig_dom}
\end{figure}

\begin{figure}[tbp]
\centering
\includegraphics[scale=0.4]{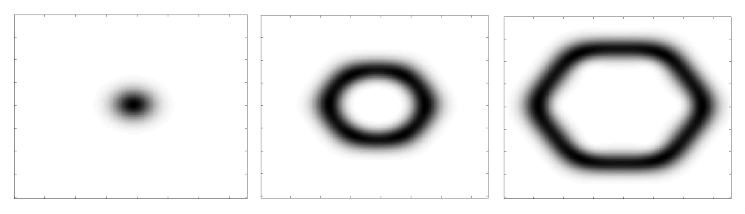}
\caption{The eigenfunctions \# 1, 7, 18 corresponding to the operator in
Fig. \protect\ref{fig_dom}}
\end{figure}

We now look into the concentrated process $\mathcal{X}^{g,\mathrm{con}%
}_{\Omega}$ introduced in Step 1. The Toeplitz operator $M^g_\Omega$ is not
a projection. However, the corresponding DPP can be realized as a random
mixture of DPP's associated with projection kernels \cite[Theorem 4.5.3]%
{DetPointRand}. Indeed, if $I_j \sim \mbox{Bernoulli}(\lambda_{j}^{\Omega })$
are independent (taking the value $1$ or $0$ with probabilities $%
\lambda_{j}^{\Omega }$ and $1-\lambda_{j}^{\Omega }$ respectively), then $%
\mathcal{X}^{g,\mathrm{con}}_{\Omega}$ is generated by the kernel
corresponding to the random operator
\begin{equation}
M_{\Omega}^{g, \mathrm{ran}} = \sum_{j \geq 1} I_j \cdot p_{g,j}^{\Omega }
\otimes p_{g,j}^{\Omega }.
\end{equation}
Precisely, this means that one first chooses a realization of the $I_j$'s
and then a realization of the DPP with the kernel above. Because of %
\eqref{eq:l4}, the first eigenvalues $\lambda_j$ are close to $1$ and thus
the corresponding $I_j$ will most likely be $1$. Similarly, for $j \gg
\left| \Omega \right|  $, the corresponding $I_j$ will most likely be $0$.
As a finite-dimensional model for WH ensembles, we propose replacing the
random Bernoulli mixing coefficients with
\begin{eqnarray}  \label{eq_outcome}
\left\{ \begin{aligned} &1, &\mbox{ for } j \leq \ensuremath{\left| \Omega
\right| }, \\ &0, &\mbox{ for } j > \ensuremath{\left| \Omega \right| }.
\end{aligned} \right.
\end{eqnarray}

\begin{definition}
\label{def_intro_wh} Let $g \in L^2(\Rdst)$ be of norm 1 - called
\emph{the window function}, let $\Omega \subseteq \Rtdst$ with non-empty interior and
finite measure and perimeter, and let $N_{\Omega }=\left\lceil \left\vert
\Omega \right\vert \right\rceil$ the least integer greater than or equal to
the Lebesgue measure of $\Omega $. The \emph{finite Weyl-Heisenberg ensemble}
is the determinantal point process\ $\mathcal{X}^g_\Omega$ with correlation
kernel \footnote{%
We do not denote this kernel by $K^g_{\Omega }$ in order to avoid a possible
confusion with the restricted kernel ${K^g}_{|\Omega}$. Note also the
notational difference between the finite ensemble $\mathcal{X}^g_\Omega$ and
the restriction of the infinite ensemble ${\mathcal{X}^g}_{|\Omega}$.}
\begin{equation*}
K_{g,\Omega }(z,z^{\prime })=\sum_{j=1}^{N_{\Omega }}p_{g,j}^{\Omega }(z)
\overline{p_{g,j}^{\Omega }(z^{\prime })}.
\end{equation*}
\end{definition}

To illustrate the construction, consider $g(t) = 2^{1/4} e^{-\pi t^2}$ and $%
\Omega = D_R = \{ z\in \bC : |z| \leq R\} $. 
The eigenfunctions of $%
M_{D_R}^g$ are explicitly given as $p_{g,j}^{D_R}(\overline{z}) = e^{\pi i x
\xi} (\pi ^{j}/j!)^{\frac{1}{2}}z^{j} e^{-\pi |z|^2/2}$, $z=x+i\xi$. They
are independent of the radius $R$ of the disk, and choosing $R$ such that $%
\left| D_R \right|  =N$, the corresponding finite WH ensemble is precisely
the finite Ginibre ensemble given by~\eqref{finite}. This well known
fact also follows as a special case from  Corollary \ref%
{coro_ident_3}.

\subsection{Scaled limits and rates of convergence}

We now discuss how finite WH ensembles behave when the number of points
tends to infinity. Let
\begin{equation*}
\rho _{g,\Omega }(z) = K_{g,\Omega }  (z,z) = \sum_{j=1}^{N_{\Omega
}}|p_{g,j}^{\Omega }(z) |^2
\end{equation*}
be the  one-point intensity of a finite Weyl-Heisenberg ensemble, so  that
\begin{equation*}
\int _D \rho _{g,\Omega } (z) dz = \mathbb{E}\left[ \mathcal{X}^g_\Omega(D) %
\right] \,
\end{equation*}
is the expected number of points to be found in $D \subseteq {\bR^{2d}}$
(see Section \ref{sec_det}). The following describes the scaled limit of the
one-point intensities.

\begin{theorem}
\label{tl1}  Let $\Omega \subset {\mathbb{R}^{2d}}$ be compact. Then the $1$%
-point intensity of the finite Weyl-Heisenberg ensemble satisfies
\begin{equation}  \label{l5}
\rho _{g,m\Omega }(m\cdot )\longrightarrow 1_{\Omega },
\end{equation}
in $L^{1}({\mathbb{R}^{2d}})$, as $m\longrightarrow +\infty $.
\end{theorem}

\begin{figure}[tbp]
\centering
\includegraphics[scale=0.5]{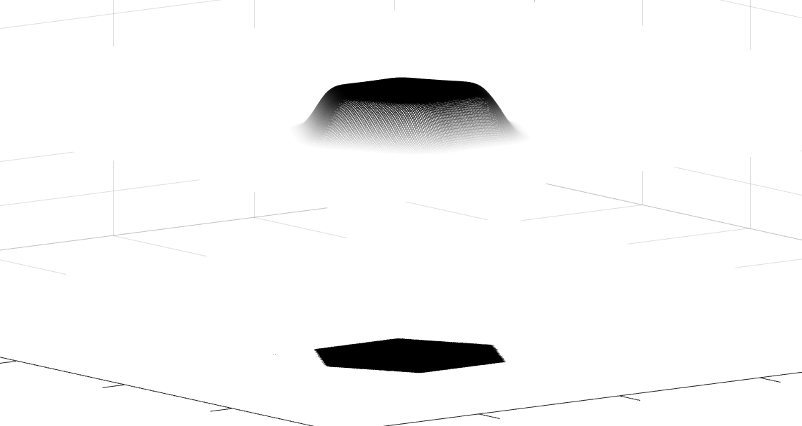}
\caption{The one-point intensity of a WH ensemble plotted over the domain in
Fig. \protect\ref{fig_dom}.}
\label{fig_onepoint}
\end{figure}

In statistical terms, the convergence in Theorem \ref{tl1} means that, as $m \longrightarrow \infty$,
\begin{equation}
\label{eq_stm}
\begin{aligned}
\tfrac{1}{m^{2d}} \mathbb{E}\left[ \mathcal{X}^g_{m\Omega}(mD) \right]
&= \tfrac{1}{m^{2d}} \int_{mD} \rho _{g,m\Omega } (z) dz
= \int_{D} \rho _{g,m\Omega } (mz) dz
\\
&\longrightarrow  \int_{D} 1_\Omega (z) dz = \abs{D \cap \Omega}.
\end{aligned}
\end{equation}

Theorem~\ref{tl1} follows immediately from \cite[Theorem 1.3]{AGR}, once the
one-point intensity $\rho_{g,\Omega}$ is recognized as the accumulated
spectrogram studied in \cite[Definition 1.2]{AGR}. We make a few
remarks as a companion to 
the illustration in Figure \ref{fig_onepoint}.

(i) When $g(t) = 2^{1/4} e^{-\pi t^2}$ and $\Omega$ is a disk of area $N$,
Theorem~\ref{tl1} follows from the circular law of the Ginibre ensemble.

(ii) The asymptotics are not restricted to disks, but hold for arbitrary
sets $\Omega $ with finite measure and also hold in arbitrary dimension, not
just for planar determinantal point processes.

(iii) The limit distribution in~\eqref{l5} is independent of the
parameterizing function $g$. This can be seen as an another instance of a
universality phenomenon~\cite{Deift, TVK, Lubinsky}.

In
view of Theorem~\ref{tl1} we will quantify the deviation of the finite WH
ensemble from its limit distribution in the $L^1$-norm, using the results in
\cite{APR}, where the sharp version of the main result in \cite{AGR} has
been obtained.

\begin{theorem}
\label{th_quant_one} Let $\rho _{g,\Omega }$ be the one-point intensity of
the finite Weyl-Heisenberg ensemble. Assume that $g$ satisfies the condition
\begin{equation}
\label{eq_mst}
\lVert g\rVert _{M^{\ast }}^{2}:=\int_{\mathbb{R}^{2d}}\left\vert
z\right\vert \left\vert \langle g, \pi (z)g\rangle \right\vert
^{2}dz<+\infty.
\end{equation}
If $\Omega$ has finite perimeter and $\left| \partial \Omega \right|
_{2d-1} \geq 1$, then
\begin{equation}  \label{l7}
\lVert \rho _{g,\Omega}-1_{\Omega}\rVert _{1}\leq C_{g} \left| \partial
\Omega \right|  _{2d-1}
\end{equation}
with a constant depending only on $\|g\|_{M^*}$.
\end{theorem}

The condition on the window $g$ in \eqref{eq_mst} amounts to mild decay in the time and frequency variables, and is satisfied by every Schwartz function. See Sections \ref{sec_change} and \ref{sec_mod} for a discussion on closely-related function classes.
The error rate in Theorem \ref{th_quant_one} is sharp - see \cite[Theorem 1.6]{APR}. Intuitively, in %
\eqref{l7} we compare the continuous function $\rho _{g,\Omega}$ with the
characteristic function $1_{\Omega}$. Thus, along every point of the
boundary of $\Omega$ (of surface measure $\left| \partial \Omega \right|
_{2d-1}$) we accumulate a pointwise error of $\mathcal{O} (1)$, leading to a
total $L^1$-error at least of order $\left| \partial \Omega \right|  _{2d-1}$%
.

\subsection{Approximation of finite polyanalytic ensembles by WH ensembles}

The second ingredient towards the proof of Theorem \ref{asy_pure} is a
comparison result that bounds the deviation between finite pure polyanalytic
ensembles and finite WH ensembles with Hermite window functions. Before
stating the result, some preparation is required. We  consider
the following  transformation, which is usually called a  gauge
transformation, 
and the change of variables $f^*(z):=f(%
\overline{z})$, $z \in \bC^d$. Given an operator $T: L^2(\Rtdst) \to L^2(%
\Rtdst)$ we denote:
\begin{align}  \label{eq_tilde_1}
\left[\widetilde{T} f \right]^* := \overline{m} \, T(f^* \, m), \qquad
m(x,\xi) := e^{-\pi i x \xi}.
\end{align}
Hence, if $T$ has the integral kernel $K$, then $\widetilde{T}$ has the
integral kernel
\begin{align}  \label{eq_tilde_2}
\widetilde{K}(z,z^{\prime})=e^{ \pi i (x^{\prime }\xi^{\prime }- x \xi)} K \left(%
\overline{z},\overline{z^{\prime }} \right), \qquad z=x+ i \xi,\, z^{\prime
}=x^{\prime }+i\xi^{\prime }.
\end{align}
(See Section \ref{sec_det} for details).
We call the operation $K \mapsto \widetilde K$ a renormalization of the
kernel $K$. With this notation, if ${K^g}$ is the kernel in %
\eqref{eq_intro_kg} and $g$ is the Gaussian window, then $\widetilde{K}_g$
is the kernel of the infinite Ginibre ensemble. In addition, the DPP's on $%
\bC^d$ associated with the kernels $K$ and $\widetilde{K}$ are related by
the transformation $z \mapsto \overline{z}$. Now, let the window $g$ be a
Hermite function
\begin{align}  \label{eq_hermite}
h_{r}(t) = \frac{2^{1/4}}{\sqrt{r!}}\left(\frac{-1}{2\sqrt{\pi}}\right)^r
e^{\pi t^2} \frac{d^r}{dt^r}\left(e^{-2\pi t^2}\right), \qquad r \geq 0.
\end{align}
The corresponding kernel $K_{h_r}$ describes (after the renormalization
above) the orthogonal projection onto the Bargmann-Fock space of pure
polyanalytic functions of type $r$ (see Section \ref{app_poly}).

Let us consider a Toeplitz operator on $L^2(\mathbb{R}^2)$ with a circular
domain $\Omega = D_R$. By means of an argument based on phase-space
symmetries (more precisely, the symplectic covariance of Weyl's
quantization) we show in Section \ref{sec_hd} that the eigenfunctions $%
\{\widetilde p_{h_r,j}^{D_R}: j \geq 1\}$ of $\widetilde{M}^{h_r}_{D_R}$ are
the normalized complex Hermite polynomials $H_{j,r}(z,\bar{z}) e^{-\frac{\pi%
}{2} |z|^2}$. In particular, as with the Ginibre ensemble, the
eigenfunctions are independent of the radius $R$. Choosing $R$ such that $%
N_{D_R} = N$, and recalling that we order the eigenvalues of $M_{D_R}^{h_r}$
by magnitude, we obtain a map $\sigma: \Nst _0 \to \Nst_0$, such that
\begin{align*}
\widetilde p_{h_r,j}^{D_R}=H_{\sigma(j),r}(z,\bar{z}) e^{-\frac{\pi}{2}
|z|^2}.
\end{align*}
Thus, the finite WH ensemble associated with $h_r$ and $D_R$ is a planar
Hermite ensemble, with correlation kernel
\begin{equation}  \label{eq_a}
\widetilde K_{h_r,D_{R}}(z,z^{\prime})=e^{ -\frac{\pi}{2} (\left\vert z\right\vert
^{2}+\left\vert z^{\prime }\right\vert ^{2})}\sum_{j=1} ^{N_{D_{R}} }
H_{\sigma (j),r}(z,\overline{z}) \overline{H_{\sigma (j),r}(z^{\prime },%
\overline{z^{\prime }})}.
\end{equation}
Comparing the correlation kernels of the finite pure polyanalytic ensemble %
\eqref{purekernel} with the finite (renormalized) WH ensemble with a Hermite
window \eqref{eq_a}, we see that in each case different subsets of the
complex Hermite basis intervene: in one case functions are ordered according
to their Hermite index, while in the other they are ordered according to the
magnitude of their eigenvalues.

Figure \ref{fig_val} shows the eigenvalues of $\widetilde{M}^{h_1}_{D_R}$,
as a function of $R$, corresponding to the eigenfunctions $H_{0,1}(z,%
\overline{z}) e^{-\tfrac{\pi}{2}\left| z \right|  ^2}$ and $H_{1,1}(z,%
\overline{z}) e^{-\tfrac{\pi}{2}\left| z \right|  ^2}$. For small values of $%
R>0$, the eigenvalue corresponding to $H_{1,1}$ is bigger than the one
corresponding to $H_{1,0}$, and thus for small $N$, the kernels in %
\eqref{purekernel} and \eqref{eq_a} do not coincide. The following result
shows that this difference is asymptotically negligible.

\begin{theorem}
\label{th_main} Let $N \in \Nst$ and $R>0$ be such that $N_{D_R}=%
\ceil{\abs{D_R}}=N$. Let $K_{h_r, D_{R}}$ be the correlation kernel of the
finite Weyl-Heisenberg ensemble associated with the Hermite window $h_r$ and
the disk $D_{R}$, and $K_{r,N}$ the correlation kernel of the $(r,N)$-pure
polyanalytic ensemble  given by \eqref{purekernel}. Then
\begin{align*}
\bignorm{\widetilde K_{h_r, D_{R}} - K_{r,N}}_{S^1} \lesssim \left| \partial
D_{R} \right|  _{1} \asymp \sqrt{N},
\end{align*}
where $\norm{\cdot}_{S^1}$ denotes the trace-norm of the corresponding
integral operators.
\end{theorem}

Since $\norm{K_{h_r, D_{R}}}_{S^1} = \norm{K_{r,N}}_{S^1} = N$, the finite
pure polyanalytic ensemble - defined by a lexicographic criterion - is
asymptotically equivalent to a finite WH ensemble - defined by optimizing
phase-space concentration. To derive Theorem \ref{th_main}, we resort to
methods from harmonic analysis on phase space. More precisely, we will use
Weyl's correspondence and account for the difference between \eqref{eq_a}
and \eqref{purekernel} as the error introduced by using two different
variants of Berezin's quantization rule (anti-Wick calculus).

\begin{figure}[tbp]
\centering
\includegraphics[scale=0.4]{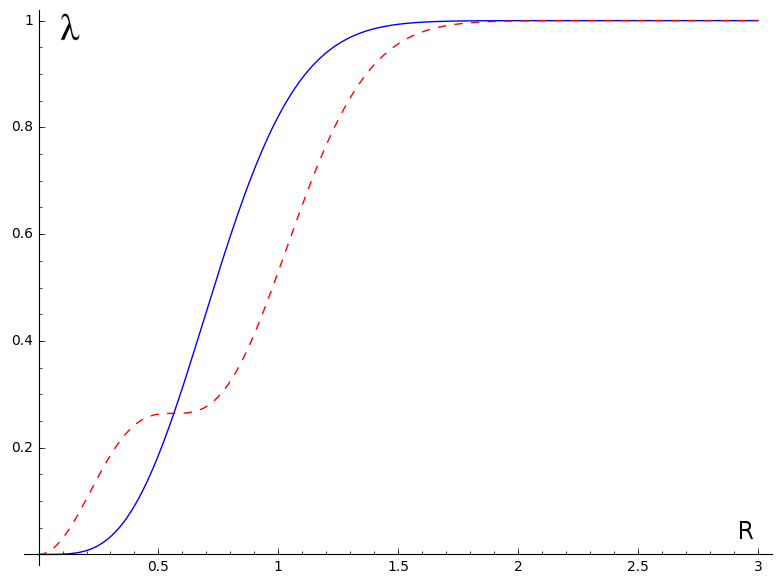}
\caption{A plot of the eigenvalues $\protect\lambda=\protect\widetilde{M}%
^{h_1}_{D_R} \left(H_{j,1}(z,\overline{z}) e^{-\tfrac{\protect\pi}{2}\left|
z \right|  ^2} \right)$, as a function of $R$, corresponding to $j=0$ (blue,
solid) and $j=1$ (red, dashed)}
\label{fig_val}
\end{figure}

Finally, Theorem \ref{asy_pure} follows by combining the comparison result
in Theorem \ref{th_main} with the asymptotics in Theorem \ref{th_quant_one}
applied to Hermite windows - see Section \ref{sec_xxx}. This argument is
reminiscent of Lubinsky's localization principle \cite{Lubinsky} that
concerns deviations between kernels of orthogonal polynomials. In the
present context, the difference between the two kernels does not stem from
an order relation between two measures, but from a permutation of the basis
functions.

\subsection{Simultaneous observability}

The independence of the eigenfunctions of $M_{D_R}^{h_r}$ of the radius $R$
yields another property of the (finite and infinite) $r$-pure polyanalytic
ensembles.

\begin{theorem}
\label{th_intro_sym} The restrictions $\{ p_{h_r,j} \big| _{D_R} : j\in \bN %
\}$ are orthogonal on $L^2(D_R)$ for all $R>0$. In the terminology of
determinantal point processes this means that the family of disks $\{D_R:
R>0\}$ is simultaneously observable for all $r$-pure polyanalytic ensembles.
\end{theorem}

This recovers and slightly extends a result of Shirai \cite{SHIRAI}. As an
application, we obtain an extension of Kostlan's theorem \cite{Kostlan} on
the absolute values of the points of the Ginibre ensemble of dimension $N$.

\begin{theorem}
\label{eq_thm4} The set of absolute values of the points distributed
according to the $r$-pure polyanalytic Ginibre ensemble has the same
distribution as $\{Y_{1,r},\ldots,Y_{n,r}\}$, where the $Y_{j}$'s are
independent and have density
\begin{equation*}
f_{Y_{j}}(x):=2 \frac{\pi^{j-r+1} r!}{j!} x^{2(j-r)+1}\left[ L_{r}^{j-r}(\pi
x^2)\right]^{2}e^{-\pi x^2},
\end{equation*}
where $L_{j}^{\alpha}$ are the Laguerre polynomials of~\eqref{eq_lag}.
(Hence, $Y^2_j$ is distributed according to a generalized Gamma 
function with density $f_{Y_j^2}(x) = \frac{\pi^{j-r+1} r!}{j!} x^{j-r}\left[ L_{r}^{j-r}(\pi
x)\right]^{2}e^{-\pi x}$).
\end{theorem}

\subsection{Organization}

Section \ref{sec_ps} presents tools from phase-space analysis, including the
short-time Fourier transform and Weyl's correspondence. Section \ref{sec_wh}
studies finite WH ensembles and more technical variants required for the
identification of finite polyanalytic ensembles as WH ensembles with Hermite
windows. This identification is carried out in Section \ref{sec_hd} by means
of symmetry arguments. The approximate identification of finite polyanalytic
ensembles with finite WH ensembles is finished in Section \ref{sec_comp} and
gives a comparison of the processes defined by truncating the complex
Hermite expansion on the one hand, and by the abstract concentration and
spectral truncation method on the other. We explain the deviation between
the two ensembles as stemming from two different quantization rules. The
proof resorts to a Sobolev embedding for certain symbol classes known
as modulation spaces. Some of the technical details are postponed to the
appendix. Theorem \ref{asy_pure} is proved in Section \ref{sec_comp}. In
Section \ref{sec_do} we apply the symmetry argument from Section \ref{sec_hd}
to rederive the so-called simultaneous observability of polyanalytic
ensembles. We also clarify the relation between the spectral expansions of
the restriction and Toeplitz kernels. Finally, the appendix provides some
background material on determinantal point processes, a certain symbol class
for pseudo-differential operators, functions of bounded variation, and
polyanalytic spaces.

\section{Harmonic analysis on phase space}

\label{sec_ps} In this section we briefly discuss our tools. These  methods
from harmonic analysis are new in the study of determinantal point processes.

\subsection{The short-time Fourier transform}

Given a window function $g\in L^{2}({\mathbb{R}^{d}})$, the short-time
Fourier transform of $f \in L^2(\Rdst)$ is
\begin{equation}  \label{eq_stft}
V_{g}f(x,\xi )=\int_{\Rdst} f(t) \overline{g(t-x)} e^{-2\pi i \xi t} dt,
\qquad (x,\xi) \in \Rtdst.
\end{equation}
The short-time Fourier transform is closely related to the Schr\"odinger
representation of the Heisenberg group, which is implemented by the
operators
\begin{equation*}
T(x,\xi ,\tau )g(t)=e^{2\pi i\tau }e^{-\pi ix\xi }e^{2\pi i\xi t}g(t-x),
\qquad (x, \xi) \in \Rdst, \tau \in \Rst.
\end{equation*}
The corresponding representation coefficients are
\begin{equation*}
\left\langle f,T(x,\xi ,\tau )g\right\rangle =e^{-2\pi i\tau }e^{\pi ix\xi
}\left\langle f,e^{2\pi i\xi \cdot}g(\cdot-x)\right\rangle = e^{-2\pi i\tau
} e^{\pi ix\xi}V_g f(x,\xi).
\end{equation*}
As the variable
$\tau $ occuring in 
the Schr\"odinger representation is unnecessary for DPPs, we will only
use  the short-time Fourier transform. We identify a pair $(x,\xi 
)\in {\mathbb{R}^{2d}}$ with the complex vector $z=x+i\xi \in {\mathbb{C}^{d}%
}$. In terms of the phase-space shifts in \eqref{eq:l2}, the short-time
Fourier transform is $V_{g}f(z):=\left\langle f,\pi (z)g\right\rangle$. The
phase-space shifts satisfy the commutation relations
\begin{equation}
\pi (x,\xi )\pi (x^{\prime },\xi ^{\prime })=e^{-2\pi i\xi ^{\prime }x}\pi
(x+x^{\prime },\xi +\xi ^{\prime }),\qquad (x,\xi ),(x^{\prime },\xi
^{\prime })\in {\mathbb{R}^{d}}\times {\mathbb{R}^{d}},  \label{com}
\end{equation}
and the short-time Fourier transform satisfies the following \emph{%
orthogonality relations}~\cite[Proposition 1.42]{folland89} \cite[Theorem
3.2.1]{Charly},
\begin{equation}
\left\langle V_{g_{1}}f_{1},V_{g_{2}}f_{2}\right\rangle _{L^2({\bR^{2d}} )}
=\left\langle f_{1},f_{2}\right\rangle _{L^2(\bR^d )} \overline{\left\langle
g_{1},g_{2}\right\rangle }_{L^2(\bR^d )}.  \label{eq_ortrel}
\end{equation}
In particular, when $\lVert g\rVert _{2}=1$, the map $V_{g}$ is an isometry
between $L^{2}({\mathbb{R}^{d}})$ and a closed subspace of $L^{2}({\mathbb{R}%
^{2d}})$:
\begin{equation} \label{isom9}
\lVert V_{g}f\rVert _{{L^{2}({\mathbb{R}^{2d}})}}=\lVert f\rVert_{{L^{2}({%
\mathbb{R}^{d}})}},\qquad f\in {L^{2}({\mathbb{R}^{d}})}.
\end{equation}
The commutation rule (\ref{com}) implies the following formula
for the short-time Fourier transform:
\begin{equation*}
V_{g}(\pi (x,\xi )f)(x^{\prime },\xi ^{\prime })=e^{-2\pi ix(\xi ^{\prime
}-\xi )}V_{g}f(x^{\prime }-x,\xi ^{\prime }-\xi ),\qquad (x,\xi ),(x^{\prime
},\xi ^{\prime })\in {\mathbb{R}^{d}}\times {\mathbb{R}^{d}}.
\end{equation*}
Since the phase-space shift of $f$ on $\mathbb{R}^d$  corresponds to a  phase-space shift of $V_gf$ on
$\mathbb{R}^{2d}$, this formula is usually called the \emph{covariance
  property} of the short-time Fourier transform.

\subsection{Special windows}

If we choose the Gaussian function $h_{0}(t)=2^{\frac{1}{4}}e^{-\pi t^{2}}$,
$t \in \Rst$, as a window in \eqref{eq_stft}, then a simple calculation
shows that
\begin{equation}
e^{-i \pi x\xi + \frac{\pi}{2} \left\vert z\right\vert
^{2}}V_{h_{0}}f(x,-\xi )=2^{1/4} \int_{\mathbb{R}}f(t)e^{2\pi tz-\pi t^{2}-%
\frac{\pi }{2}z^{2}}dt=\mathcal{B}f(z),  \label{Bargmann}
\end{equation}
where $\mathcal{B}f(z)$ is the \emph{Bargmann transform} of $f$ \cite%
{Bargmann}, \cite[Chapter~1.6]{folland89}. The Bargmann transform $\mathcal{B%
}$ is a unitary isomorphism from $L^{2}(\bR)$ onto the Bargmann-Fock space $%
\mathcal{F}(\bC)$ consisting of all entire functions satisfying
\begin{equation}
\left\Vert F\right\Vert _{\mathcal{F}(\mathbb{C})}^{2}= \int_{\mathbb{C}%
}\left\vert F(z)\right\vert ^{2}e^{-\pi \left\vert z\right\vert
^{2}}dz<\infty .  \label{Focknorm}
\end{equation}
We now explain the relation between polyanalytic Fock spaces and
phase-space  analysis with Hermite windows $\{h_r: r\geq0\}$. The $r$-%
\emph{pure polyanalytic Bargmann transform} \cite{Abreusampling} is the map $%
\mathcal{B}^{r}:L^{2}({\mathbb{R}})\rightarrow L^{2}({\mathbb{C}},e^{-\pi
\left\vert z\right\vert ^{2}})$
\begin{equation}  \label{eq_relation}
\mathcal{B}^{r}f(z):=e^{-i\pi x\xi +\tfrac{\pi }{2}\left\vert z\right\vert
^{2}}V_{h_{r}}f(x,-\xi ),\qquad z=x+i\xi .
\end{equation}
This map defines an isometric isomorphism between $L^{2}({\mathbb{R}})$ and
the pure polyanalytic-Fock space $\mathcal{F}^{r}(\mathbb{C})$ (see Section %
\ref{sec_pure_pol}). The orthogonality relations \eqref{eq_ortrel} show that
for $r\not=r^{\prime }$, $V_{h_{r}}f_{1}$ is orthogonal to $V_{h_{r^{\prime
}}}f_{2}$ for all $f_1,f_2 \in L^{2}({\mathbb{R}})$. The relation between
phase-space  analysis and polyanalytic functions discovered in \cite%
{Abreusampling} can be understood in terms of the \emph{Laguerre connection}
\cite[Chapter 1.9]{folland89}:
\begin{equation}  \label{l9b}
V_{h_{r}} h_j (x,-\xi ) = e^{i\pi x\xi -\tfrac{\pi }{2}\left| z \right|  ^2}
H_{j,r}(z,\bar{z}),
\end{equation}
which, in terms of the polyanalytic Bargmann transform reads as
\begin{equation}  \label{l9}
\mathcal{B}^{r}h_{j}(z)=H_{j,r}(z,\bar{z}),
\end{equation}
see also \cite{Abreusampling}.

\subsection{The range of the short-time Fourier transform}

For $\lVert g\rVert _{2}=1$, the short-time Fourier transform $V_{g}$
defines an isometric map $V_{g}:{L^{2}({\mathbb{R}^{d}})}\rightarrow {L^{2}({%
\mathbb{R}^{2d}})}$ with range
\begin{equation*}
{\mathcal{V}_{g}}:=\big\{\,V_{g}f\,:\,f\in {L^{2}({\mathbb{R}^{d}})}\,\big\} %
\subseteq {L^{2}({\mathbb{R}^{2d}})}.
\end{equation*}
The adjoint of $V_{g}$ can be written formally as $V_{g}^{\ast }:L^{2}({%
\mathbb{R}^{2d}})\rightarrow L^{2}({\mathbb{R}^{d}})$,
\begin{equation*}
V_{g}^{\ast }F=% \int_{{\mathbb{R}^{d}}\times {\mathbb{R}^{d}}}F(x,\xi
% )g(t-x)e^{2\pi i\xi t}dxd\xi =
\int _{{\bR^{2d}} } F(z) \pi (z)g \, dz
,\qquad t\in {\mathbb{R}^{d}} \, ,
\end{equation*}
where the integral is to be taken as a vector-valued integral. 
The orthogonal projection $P_{\mathcal{V}_{g}}:L^{2}({\mathbb{R}^{2d}}
)\rightarrow {\mathcal{V}_{g}}$ is then $P_{\mathcal{V}_{g}}=V_{g}V_{g}^{
\ast }$. Explicitly, $P_{\mathcal{V}_{g}}$ is the integral operator
\begin{equation*}
P_{\mathcal{V}_{g}}F(z)=\int_{\mathbb{R}^{2d}} {K^g}(z,z^{\prime })
F(z^{\prime }) dz^{\prime },\qquad z=(x,\xi )\in {\mathbb{R}^{2d}},
\end{equation*}
where the \emph{reproducing kernel} ${K^g}$ is given by \eqref{eq:l1}. Every
function $F\in {\mathcal{V}_{g}}$ is continuous and satisfies the
reproducing formula $F(z)=\int _{{\bR^{2d}} } F(z^{\prime }){K^g}%
(z,z^{\prime })dz^{\prime }$.

\subsection{Metaplectic rotation}

We will make use of a rotational symmetry argument in phase space. Let $%
R_{\theta }:= \big[%
\begin{smallmatrix}
\cos (\theta ) & -\sin (\theta ) \\
\sin (\theta ) & \cos (\theta )%
\end{smallmatrix}
\big]$ denote the rotation by the angle $\theta \in {\mathbb{R}}$. The \emph{%
metaplectic rotation} is the operator given in the Hermite basis $%
\left\{h_{r}:r\geq 0\right\} $ by
\begin{equation}
{\mu(R_\theta)}f=\sum_{r\geq 0}e^{i r\theta }\left\langle
f,h_{r}\right\rangle h_{r},\qquad f\in L^{2}(\mathbb{R}) \, ,
\label{rotation}
\end{equation}
in particular, $\mu (R_\theta ) h_r = e^{i r\theta } h_r$. The standard and
metaplectic rotations are related by
\begin{equation}
V_{g}f(R_{\theta }(x,\xi ))=e^{\pi i(x\xi -x{^{\prime }}\xi {^{\prime }}%
)}V_{\mu(R_{-\theta}) g}\mu(R_{-\theta}) f(x,\xi ),\mbox{ where }(x{^{\prime
}},\xi {^{\prime }})=R_{\theta }(x,\xi ).  \label{eq_cov}
\end{equation}
This formula is a special case of the symplectic covariance of the
Schr\"odinger representation; see \cite[Chapters 1 and 2]{folland89}, \cite[%
Chapter 9]{Charly}, or \cite[Chapter 15]{dego11}) for background and proofs.

\subsection{Time-frequency localization and Toeplitz operators}

\label{sec_toep} Let us consider $g$ with $\lVert g\rVert _{2}=1$. For $m\in
L^{\infty }({\mathbb{R}^{2d}})$, the \emph{Toeplitz operator} $M_{m}^{g}:{%
\mathcal{V}_{g}}\rightarrow {\mathcal{V}_{g}}$ is
\begin{equation*}
M_{m}^{g}F:=P_{\mathcal{V}_{g}}(m\cdot F),\qquad F\in {\mathcal{V}_{g}},
\end{equation*}
and its integral kernel at a point $(z,z')$ is given by 
\begin{align}\label{eq_ker_new}
K_m(z,z^\prime)= \int_{\Rtdst} {K^g} (z,z^{\prime
\prime }) m (z^{\prime \prime }) {K^g} (z^{\prime \prime },z^{\prime
}) \, dz^{\prime \prime }.
\end{align}
When $m=1_\Omega$, the last expression coincides with \eqref{Toe}. (The operator $M_{m}^{g}$
is defined on $\mathcal{V}_g$; the kernel in \eqref{eq_ker_new} represents the
extension of $M_{m}^{g}$ to $L^2(\Rtdst)$ that is $0$ on $\mathcal{V}_g^\perp
$.) Clearly, $\lVert M_{m}^{g}\rVert _{{\mathcal{V}_{g}}\rightarrow {%
\mathcal{V}_{g}}}\leq \lVert m\rVert _{\infty }$. In addition, it is easy to
see that if $m\geq 0$, then $M_{m}^{g}$ is a positive operator. If
$m\in L^1(\rdd )$, then $M_m^g$ is trace-class.  By \eqref{eq_ker_new} the trace of $M_m^g$ is
\begin{equation}
  \label{eq:uura}
  \mathrm{trace} (M_m^g ) = \int_\rdd   K_m(z,z) \, dz = \int _{\rdd}
  \int _{\rdd } |K^g(z,z^{\prime \prime})|^2 m(z^{\prime \prime}) \, dz
  dz^{\prime \prime} = \int _{\rdd } m(z^{\prime \prime}) \,
  dz^{\prime \prime} \, ,
\end{equation}
because the isometry property \eqref{isom9} implies that 
$$
\int _{\rdd } |K^g(z,z^{\prime \prime})|^2 \, dz = \int _{\rdd } |\langle
\pi (z^{\prime \prime})g, \pi (z) g\rangle |^2 \, dz = 1 \, .
$$
The
\emph{time-frequency  localization operator} with window $g$ and
symbol $m$ is $% 
H_{m}^{g}:=V_{g}^{\ast }M_{m}^{g}V_{g}:{L^{2}({\mathbb{R}^{d}})}\rightarrow {%
L^{2}({\mathbb{R}^{d}})}$. Hence $M_{m}^{g}$ and $H_{m}^{g}$ are unitarily
equivalent. \footnote{%
The operator $H_{m}^{g}$ should not be confused with the complex Hermite
polynomial $H_{j,r}$.} The situation is depicted in the following diagram.
\begin{align}  \label{eq_diagram}
\xymatrix{ L^2(\rd) \ar[d]_{V_g} \ar[r]^{H^g_m} &L^2(\rd) \ar[d]^{V_g} \\
\gabspace_g \ar[rd]_{m \cdot} \ar[r]^{M^g_m} & \gabspace_g \\ & L^2(\rdd)
\ar[u]_{P_{\gabspace_g}} }
\end{align}

Explicitly, the time-frequency localization operator applies a mask to the
short-time Fourier transform:
\begin{equation*}
H_{m}^{g}f:=\int_{{\mathbb{R}^{2d}}}m(z)V_{g}f(z)\pi (z)g\,dz, \qquad f \in
L^2({\bR^{2d}}).
\end{equation*}
As we will use the connection between time-frequency localization on $\bR^d $
and Toeplitz operators on ${\bR^{2d}} $ in a crucial argument, we write ~%
\eqref{eq_diagram} as a formula
\begin{align}
\langle H_m^g f, u\rangle &= \langle V_g (V_g^* M_m^g V_gf), V_gu\rangle
\notag \\
&= \langle P_{\mathcal{V}_g} (m \, V_gf), V_gu\rangle  \notag \\
&= \langle m \, V_gf ,V_gu \rangle \, .  \label{eq:o12}
\end{align}
This formula makes sense for $f,u\in L^2(\bR^d)$ and $m\in L^\infty({\bR^{2d}%
} )$, but also for many other assumptions~\cite{cogr03}.

Time-frequency localization operators are useful in signal processing because they model
time-varying filters. For Gaussian windows, they have been studied in signal
processing by Daubechies \cite{Daubechies} and as Toeplitz operators on
spaces of analytic functions by Seip \cite{Seip}; see also \cite[Section 1.4]%
{AGR}. When $m\in L^{1}({\mathbb{R}^{2d}})$, $H_{m}^{g}$ is
trace-class by \eqref{eq:uura}  and
\begin{equation}  \label{eq:o9}
\mathrm{trace}(H_{m}^{g})=\int_{{\bR^{2d}} } m(z) dz \, . 
\end{equation}
For more details see  \cite{heil1, heil2, cogr03}.
When $%
m=1_{\Omega }$, the indicator function of a set $\Omega $, we write $%
M_{\Omega }^{g}$ and $H_{\Omega }^{g}$. In this case, the positivity
property implies that $0\leq M_{\Omega }^{g}\leq I$.

\subsection{The Weyl correspondence}

The \emph{Weyl transform} of a distribution $\sigma \in \mathcal{S}^{\prime
}({\mathbb{R}^{d}}\times {\mathbb{R}^{d}})$ is an operator $\sigma ^{w}$
that is formally defined on functions $f:{\mathbb{R}^{d}}\rightarrow {%
\mathbb{C}}$ as
\begin{equation*}
\sigma ^{w} f(x):=\int_{{\mathbb{R}^{d}}\times {\mathbb{R}^{d}}}\sigma
\left( \frac{x+y}{2},\xi \right) e^{2\pi i(x-y)\xi }f(y)dyd\xi ,\qquad x\in {%
\mathbb{R}^{d}}.
\end{equation*}
Every continuous linear operator $T:\mathcal{S}({\mathbb{R}^{d}})\rightarrow
\mathcal{S}^{\prime }({\mathbb{R}^{d}})$ can be represented in a unique way
as $T=\sigma ^{w}$, and $\sigma $ is called its \emph{Weyl symbol} (see \cite%
[Chapter 2]{folland89}). The Wigner distribution of a test function $g\in
\mathcal{S}({\mathbb{R}^{d}})$ and a distribution $f\in \mathcal{S}^{\prime
}({\mathbb{R}^{d}})$ is
\begin{equation*}
W(f,g)(x,\xi )=\int_{{\mathbb{R}^{2d}}}f(x+\tfrac{t}{2})\overline{g(x-\tfrac{%
t}{2})}e^{-2\pi it\xi }dt.
\end{equation*}
The integral has to be understood distributionally. The map $(f,g)\mapsto
W(f,g)$ extends to other function classes, for example, for $f,g\in L^{2}({%
\mathbb{R}^{d}})$, $W(f,g)$ is well-defined and
\begin{equation}  \label{eq_norm_wigner}
\norm{\wigfg}_2 = \norm{f}_2 \norm{g}_2.
\end{equation}
The Wigner distribution is closely related to the short-time Fourier
transform:
\begin{equation*}
W(f,g)(x,\xi )=2^{d}e^{4\pi ix\cdot \xi }V_{\tilde{g}}f(2x,2\xi ),
\end{equation*}
where $\tilde{g}(x)=g(-x)$. The action of $\sigma ^{w}$ on a distribution
can be easily described in terms of the Wigner distribution:
\begin{equation*}
\left\langle \sigma ^{w}f,g\right\rangle =\left\langle
\sigma,W(g,f)\right\rangle .
\end{equation*}
Time-frequency localization operators have the following simple description
in terms of the Weyl calculus:
\begin{equation}  \label{eq_tf_weyl}
H_{m}^{g}=\left( m\ast W(g,g) \right) ^{w}.
\end{equation}

\section{Finite Weyl-Heisenberg ensembles}

\label{sec_wh}

\subsection{Definitions}

To define finite Weyl-Heisenberg processes, we consider a domain $\Omega
\subseteq {\mathbb{R}^{2d}}$ with non-empty interior, finite measure and
finite perimeter, i.e., the characteristic function of $\Omega $ has bounded
variation (see Section \ref{sec_bv}). Since $M_{\Omega }^{g}$ is
trace-class, the Toeplitz operator $M_{\Omega }^{g}$ can be diagonalized as
\begin{equation}  \label{eq_diag_toep}
M_{\Omega }^{g}=\sum_{j\geq 1}\lambda_{j}^{\Omega }\, p_{g,j}^{\Omega }
\otimes p_{g,j}^{\Omega },\qquad f\in {L^{2}({\mathbb{R}^{2d}})},
\end{equation}
where $\left\{ {\lambda_{j}^{\Omega }}:j\geq 1\right\} $ are the non-zero
eigenvalues of $M_{\Omega }^{g}$ in decreasing order and the corresponding
eigenfunctions $\left\{ p_{g,j}^{\Omega }:j\geq 1\right\} $ are normalized
in $L^{2}$. The operator $M_{\Omega }^{g}$ may have a non-trivial kernel,
but it is known that it always has infinite rank \cite[Lemma 5.8]{doro14},
therefore, the sequences $\{ {\lambda_{j}^{\Omega }}:j\geq 1\} $ and $%
\{p_{g,j}^{\Omega }: j \geq 1\}$ are indeed infinite. In addition, as
follows from \eqref{eq:o9}, we have
\begin{equation}  \label{eq:o10}
0\leq {\lambda_{j}^{\Omega }}\leq 1,\mbox{ and }\sum_{j \geq 1}{%
\lambda_{j}^{\Omega }}=\left\vert \Omega \right\vert.
\end{equation}
We remark that the eigenvalues $\lambda_{j}^{\Omega }$ do depend on the
window function $g$. When we need to stress this dependence we write $%
\lambda_{j}(\Omega,g)$.

\emph{The finite Weyl-Heisenberg ensemble} $\mathcal{X}^g_\Omega$ is given
by Definition \ref{def_intro_wh}. For technical reasons, we will also
consider a more general class of WH ensembles depending on an extra
ingredient. Given a subset $I\subseteq {\mathbb{N}}$, we let $\mathcal{X}%
^g_{\Omega,I}$ be the determinantal point process with correlation kernel
\begin{equation*}
K_{g,\Omega ,I}(z,z^{\prime })=\sum_{j\in I}p_{g,j}^{\Omega }(z)\overline{%
p_{g,j}^{\Omega }(z^{\prime })}.
\end{equation*}
When $I=\{1,\ldots ,N_{\Omega }\}$ we obtain the finite WH ensemble $%
\mathcal{X}^g_\Omega$, while for $I={\mathbb{N}}$ we obtain the infinite
ensemble. (In the latter case, the resulting point-process is independent of
domain $\Omega $.) Later we need to analyze the properties of the ensemble $%
\mathcal{X}^g_{\Omega,I}$ with respect to  variations of the index set $I$. When no
subset $I$ is specified, we always refer to the ensemble $\mathcal{X}%
^g_\Omega$ associated with $I=\{1,\ldots,N_{\Omega }\}$.

\begin{rem}
\rm{The process $\mathcal{X}^g_{\Omega,I}$ is well-defined due to the
Macchi-Soshnikov theorem (see Section \ref{sec_det}). Indeed, since the
kernel $K_{g,\Omega ,I}$ represents an orthogonal projection, we only need
to verify that it is locally trace-class. This follows easily from the facts
that $0 \leq K_{g,\Omega ,I}(z,z) \leq {K^g}(z,z)=1$ and that the
restriction operators are positive (see Section \ref{sec_duality}).}
\end{rem}

\subsection{Universality and rates of convergence}

The one-point intensity associated with a Weyl-Heisenberg ensemble $\mathcal{%
X}^g_{\Omega,I}$ is
\begin{equation*}
\rho _{g,\Omega,I}(z):=\sum_{j \in I}\left\vert p_{g,j}^{\Omega
}(z)\right\vert ^{2}.
\end{equation*}
For $\mathcal{X}^g_\Omega$, the intensity $\rho_{g,\Omega}$ has been studied
in the realm of signal analysis, where it is known as the \emph{accumulated
spectrogram} \cite{AGR,APR}. (Another interesting connection between DPP's
and signal analysis is the completeness results of Ghosh \cite{ghosh2}.) The
results in \cite{AGR,APR} imply Theorems \ref{tl1} and \ref{th_quant_one},
which apply to the finite Weyl-Heisenberg ensembles $\mathcal{X}^g_\Omega$.
For the general ensemble $\mathcal{X}^g_{\Omega,I}$ we have the following
lemma.

\begin{lemma}
\label{lemma_one_point_I} Let $\rho_{g,\Omega,I}$ be the one-point intensity
of a WH ensemble $\mathcal{X}^g_{\Omega,I}$ with $\#I < \infty$. Then
\begin{align*}
\lVert \rho_{g,\Omega,I}-1_\Omega\rVert_{L^1({\mathbb{R}^{2d}})} =
\#I-\left| \Omega \right| + 2 \sum_{j \notin I} \lambda_{j}^{\Omega }.
\end{align*}
\end{lemma}

\begin{proof}
Using that $0\leq \rho_{g,\Omega,I} \leq 1$ and \eqref{eq:o11} and %
\eqref{eq:o10}, we first calculate
\begin{align*}
\lVert \rho_{g,\Omega,I}-1_\Omega\rVert_{L^1(\Omega)}=\int_\Omega
\left(1-\rho_{g,\Omega,I}(z)\right)\, dz = {|\Omega|} - \sum_{j \in I}
\lambda_{j}^{\Omega }= \sum_{j \notin I} \lambda_{j}^{\Omega }.
\end{align*}
Second, since the eigenfunctions are normalized and $\int _\Omega
|p_{g,j}^\Omega (z)|^2 \, dz =  \lambda _j$, we have
\begin{align*}
\lVert \rho_{g,\Omega,I}-1_\Omega\rVert_{L^1({\mathbb{R}^{2d}}
\setminus\Omega)}&=\int_{{\mathbb{R}^{2d}}\setminus\Omega} \rho_{g,\Omega,I}(z) \, dz = 
\sum _{j\in I} \int _{\mathbb{R}^{2d}}
|p_{g,j}^\Omega (z)|^2 \, dz - \int _\Omega
|p_{g,j}^\Omega (z)|^2 \, dz  \\
 &     = \sum_{j \in I} \left(1-\lambda_{j}^{\Omega
}\right) = \#I - \sum_{j \in I} \lambda_{j}^{\Omega }=\#I - \left| \Omega \right| +
\sum_{j \notin I} \lambda_{j}^{\Omega }.
\end{align*}
The conclusion follows by adding both estimates.
\end{proof}

\section{Hermite windows and polyanalytic ensembles}

\label{sec_hd}

\subsection{Eigenfunctions of Toeplitz operators}

We first investigate the eigenfunctions of Toeplitz operators with Hermite
windows $\{h_r: r \geq 0\}$ and circular domains.

\begin{prop}
\label{prop_hn_disks} Let $D_R \subseteq {\mathbb{R}}^2$ be a disk centered
at the origin. Then the family of Hermite functions is a complete set of
eigenfunctions for $H^{h_r}_{D_R}$. As a consequence, the set $\{H_{j,r}(z,%
\overline{z}) e^{-\pi |z|^2/2}: j \geq 0\}$ forms a complete set of
eigenfunctions for $\widetilde M_{D_R}^{h_r}$ (where $\widetilde
M_{D_R}^{h_r}$ is related to $M_{D_R}^{h_r}$ by \eqref{eq_tilde_1}.
\end{prop}

\begin{proof}
Consider the metaplectic rotation $R_{\theta }$ with angle $\theta \in {%
\mathbb{R}}$ defined in \eqref{rotation}. For $f,u\in L^{2}({\mathbb{R}})$,
we use first \eqref{eq:o12} and then the covariance property in %
\eqref{eq_cov} and the rotational invariance of $D{_{R}}$ to compute:
\begin{align*}
\left\langle {\mu(R_\theta)}^{\ast }H_{D_R}^{h_r} {\mu(R_\theta)}%
f,u\right\rangle & =\left\langle H_{D_R}^{h_r}{\ \mu(R_\theta)}f,{%
\mu(R_\theta)}u\right\rangle =\left\langle 1_{{D_{R}}}V_{h_{r}}{\mu(R_\theta)%
}f,V_{h_{r}}{\ \mu(R_\theta)}u\right\rangle \\
& =\left\langle 1_{D{_{R}}}V_{{\mu(R_\theta)}h_{r}} \mu(R_\theta) f,V_{{%
\mu(R_\theta)}h_{r}}{\mu(R_\theta)}u\right\rangle \\
&=\left\langle 1_{D{_{R}}}V_{h_{r}}f(R_{-\theta} \,\cdot),
V_{h_{r}}u(R_{-\theta }\,\cdot)\right\rangle \\
& =\int_{D{_{R}}}V_{h_{r}}f(z)V_{h_{r}}u(z)dz =\left\langle
H_{D_R}^{h_r}f,u\right\rangle.
\end{align*}
We conclude that ${\mu(R_\theta)}^{\ast }H_{D_R}^{h_r}{\mu(R_\theta)}%
=H_{D_R}^{h_r}$, for all $\theta \in {\mathbb{R}}$. Applying this identity
to a Hermite function gives
\begin{align*}
{\mu(R_\theta)}^{\ast} H_{D_R}^{h_r} h_j &= {\mu(R_\theta)}^{\ast}
H_{D_R}^{h_r}{\mu(R_\theta)} \left( e^{-ij \theta} h_j \right) \\
&= e^{-ij \theta} {\mu(R_\theta)}^{\ast} H_{D_R}^{h_r}{\mu(R_\theta)} h_j =
e^{-ij \theta} H_{D_R}^{h_r} h_j.
\end{align*}
Thus, $H_{D_R}^{h_r} h_j$ is an eigenfunction of ${\mu(R_\theta)}^{\ast}$
with eigenvalue $e^{-ij \theta}$. For irrational $\theta$, the numbers $%
\{e^{-ij \theta}: j \geq 0\}$ are all different, and, therefore, the
eigenspaces of ${\mu(R_\theta)}^{\ast}$ are one-dimensional. Hence, $%
H_{D_R}^{h_r} h_j$ must be a multiple of $h_j$. Thus, we have shown that
each Hermite function is an eigenfunction of $H_{D_R}^{h_r}$. Since the
family of Hermite functions is complete, the conclusion follows. The
statement about the complex Hermite polynomials follows from \eqref{l9b} and %
\eqref{eq_diagram}; the extra phase-factors and conjugation bars disappear
due to the renormalization $M_{D_R}^{h_r} \mapsto \widetilde M_{D_R}^{h_r}$.
\end{proof}

\subsection{Eigenvalues of Toeplitz operators}

As a second step to identify polyanalytic ensembles as WH ensembles, we
inspect the eigenvalues of Toeplitz operators.

\begin{lemma}
\label{lemma_hn_eig} Let $R>0$. Then the eigenvalue of $H_{D_{R}}^{h_{r}}$
corresponding to $h_{j}$ and the eigenvalue of $\widetilde M_{D_R}^{h_r}$
corresponding to $H_{j,r}(z,\overline{z}) e^{-\pi |z|^2/2}$ are
\begin{equation}  \label{eq_muR}
\mu_{j,R}^{r}:= \left<H_{D_{R}}^{h_{r}} h_{j},h_j\right>
=\int_{D_{R}}\left\vert H_{r,j}(z,\bar{z})\right\vert ^{2}e^{-\pi \left\vert
z\right\vert ^{2}}dz.
\end{equation}
In particular, $\mu_{j,R}^{r}\not=0$  for all $j,r\geq 0$ and $R>0$, and
\begin{align}  \label{eq_hr}
H_{D_{R}}^{h_{r}} &= \sum_{j \geq 0} \mu^r_{j,R} \, h_j \otimes h_j.
\end{align}
\end{lemma}

\begin{proof}
\eqref{eq_muR} follows immediately  from the definitions. According to
\eqref{ComplexHermite}, $H_{r,j}$
vanishes only on a set of measure zero, thus  we
conclude that $\mu_{j,R}^{r}\not=0$. The diagonalization follows from
Proposition \ref{prop_hn_disks}.
\end{proof}

\begin{remark}
\label{rem_mult_1} \rm{ Figure \ref{fig_val} shows a plot of $\mu^1_{0,R}$
(solid, blue) and $\mu^1_{1,R}$ (dashed, red) as a function of $R$. Note
that for a certain value of $R$, the eigenvalue $\mu^1_{0,R}=\mu^1_{1,R}$ is
multiple.}
\end{remark}

\subsection{Identification as a WH ensemble}

We can now identify finite pure polyanalytic ensembles as WH ensembles.

\begin{prop}
\label{prop_ident_1} Let $J \subseteq \Nst_0$ and $R>0$, then there exist a
set $I \subseteq {\mathbb{N}}$ with $\#I=\#J$ such that
\begin{align}  \label{eq_set}
\left\{ V_{h_r} h_{j}: j \in J \right\} =\left\{p_{h_{r},j}^{D_{R}}:j\in
I\right\}.
\end{align}
\end{prop}

\begin{proof}
By Proposition \ref{prop_hn_disks} every Hermite function $h_j$ is an
eigenfunction of $H^{h_r}_{D_R}$. In addition, by Lemma \ref{lemma_hn_eig},
the corresponding eigenvalue $\mu^r_{j,R}$ is non-zero. Hence $V_{h_r} h_j$
is one of the functions $p_{h_r,j^{\prime }}^{D_R }$ in the diagonalization %
\eqref{eq_diag_toep}. The set $I := \{j^{\prime }: j \in J\}$ satisfies %
\eqref{eq_set}.
\end{proof}

As a consequence, we obtain the following.

\begin{prop}
\label{prop_ident_2} The pure polyanalytic Ginibre ensemble with kernel $%
K_{r,N}$ in \eqref{purekernel} can be identified with a finite WH ensemble
in the following way. Let $D_{R_{N}}\subset {\mathbb{C}}$ be the disk with
area $N$. Let $I_{r,N} \subseteq {\mathbb{N}}$ be a set such that
\begin{align}  \label{eq_ident_2}
\left\{ V_{h_r} h_{0}, \ldots, V_{h_r} h_{N-1} \right\}
=\left\{p_{h_{r},j}^{D_{R_N}}:j\in I_{r,N} \right\},
\end{align}
and $\# I_{r,N}=N$, whose existence is granted by Proposition \ref%
{prop_ident_1}. Then $\widetilde K_{h_r, D_{R_{N}}, I_{r,N}}=K_{r,N}$, and
the corresponding point processes coincide. In particular
\begin{align}  \label{eq_related}
\rho_{r,N}(z)=\rho_{h_{r},D_{R_{N}},I_{r,N}}(z), \qquad z \in \bC.
\end{align}
\end{prop}

\begin{proof}
Since $\# I_{r,N}=N$, we can write
\begin{align*}
K_{h_r, D_{R_{N}}}(z,z^{\prime }) = \sum_{j \in I_{r,N}}
p_{h_{r},j}^{D_{R_N}}(z) \overline{p_{h_{r},j}^{D_{R_N}}(z^{\prime })} =
\sum_{j=0}^{N-1} V_{h_r} h_{j}(z) \overline{V_{h_r} h_{j}(z^{\prime })}.
\end{align*}
Using \eqref{eq_tilde_2} and \eqref{l9b} we conclude that
\begin{align*}
\widetilde K_{h_r, D_{R_{N}}}(z,z^{\prime }) = \sum_{j=0}^{N-1} H_{j,r}(z,%
\overline{z}) e^{-\pi |z|^2/2} \overline{H_{j,r}(z^{\prime },\overline{%
z^{\prime }})} e^{-\pi |z^{\prime 2}/2} = K_{r,N}(z,z^{\prime }),
\end{align*}
as desired. This implies that the point processes corresponding to $K_{h_r,
D_{R_{N}}}$ and $K_{r,N}$ are related by transformation $z \mapsto \overline{%
z}$. Since $H_{j,r}(z,\overline{z})=\overline{H_{j,r}(\overline{z},z)}$, the
intensities of the pure $(r,N)$-polyanalytic ensemble are invariant under
the map $z \mapsto \overline{z}$ and the conclusion follows.
\end{proof}

While Proposition \ref{prop_ident_2} identifies finite pure polyanalytic
ensembles with WH ensembles in the generalized sense of Section \ref{sec_wh}%
, this is just a technical step. Our final goal is to compare finite
polyanalytic ensembles with finite WH ensembles in the sense of Definition %
\ref{def_intro_wh}, where the index set is $I_{r,N}=\{1, \ldots, N\}$.
Before proceeding we note that for the Gaussian $h_0$ such comparison is in
fact an exact identification.

\begin{coro}
\label{coro_ident_3} For $r=0$, the set $I_{0,N}$ from Proposition \ref%
{prop_ident_2} is $I_{0,N}=\{0, \ldots, N-1\}$. Thus, the $N$-dimensional
Ginibre ensemble has the same distribution as the finite WH ensemble $%
\mathcal{X}^{h_0}_{D_{R_N}}$, and
\begin{align}  \label{eq_related_3}
\rho_{0,N}(z)=\rho_{h_{0},D_{R_{N}}}(z), \qquad z \in \bC.
\end{align}
\end{coro}

\begin{proof}
The claim amounts to saying that the eigenvalues $\mu^0_{j,R}$ in %
\eqref{eq_muR} are decreasing for all $R>0$, so that the ordering of the
eigenfunctions in \eqref{eq_diag_toep} coincides with the indexation of the
complex Hermite polynomials. The explicit formula in \eqref{eq_muR} in the
case $r=0$ gives the sequence of \emph{incomplete Gamma functions}:
\begin{align*}
\mu^0_{j,R} = \frac{1}{j!} \int_0^{\pi R^2} t^j e^{- t} dt = 1- e^{-\pi R^2}
\sum_{k=0}^{j} \frac{\pi^k}{k!} R^{2k},
\end{align*}
which is decreasing in $j$ (see for example \cite[Eq. 6.5.13]{special}).
\end{proof}

\section{Comparison between finite WH and polyanalytic ensembles}

\label{sec_comp} Having identified finite pure polyanalytic ensembles as WH
ensembles associated with a certain subset of eigenfunctions $I$, we now
investigate how much this choice deviates from the standard one $I=\{1,
\ldots, N\}$. Thus, we compare finite pure polyanalytic ensembles to the
finite WH ensembles of Definition \ref{def_intro_wh}.

\subsection{Change of quantization}

\label{sec_change} As a main technical step, we show that the change of the
window of a time-frequency localization operator affects the distribution of
the corresponding eigenvalues in a way that is controlled by the perimeter
of the localization domain. When $g$ is a Gaussian, the map $m\mapsto
H_{m}^{g}$ is called \emph{Berezin's quantization} or \emph{anti-Wick
calculus} \cite[Chapter 2]{folland89} or \cite{lerner03}. The results in
this section show that if Berezin's quantization is considered with respect
to more general windows and in ${\bR^{2d}} $, the resulting calculus enjoys
similar asymptotic spectral properties. We consider the function class
\begin{equation}  \label{eq_def_m1}
{M^{1}}({\mathbb{R}^{d}}):=\big\{\,f\in L^{2}({\mathbb{R}^{d}})\,:\,\lVert
f\rVert _{{M^{1}}}:=\lVert V_{\phi }f\rVert _{L^{1}({\mathbb{R}^{2d}}
)}<+\infty \,\big\},
\end{equation}
where $\phi (x)=2^{d/4}e^{-\pi \left\vert x\right\vert ^{2}}$. The class $%
M^{1}$ is one of the \emph{modulation spaces} used in signal processing. It
is also important as a symbol-class for pseudo-differential operators.
Indeed, the following lemma, whose proof can be found in \cite{gr96}, gives
a trace-class estimate in terms of the $M^{1}$-norm of the Weyl symbol (see
also \cite{heil1, heil2, cogr03}).

\begin{prop}
\label{prop_s1} Let $\sigma \in {M^1}({\mathbb{R}^{2d}})$. Then $\sigma ^{w}$
is a trace-class operator and
\begin{equation*}
\lVert \sigma ^{w}\rVert _{S^{1}}\lesssim \lVert \sigma \rVert _{{M^1}},
\end{equation*}
where $\norm{\cdot}_{S^{1}}$ denotes the trace-norm.
\end{prop}

The next lemma will allow us to exploit cancellation properties in the $M^1$%
-norm. Its proof is postponed to Section \ref{sec_mod}.

\begin{lemma}[A Sobolev embedding for $M^1$]
\label{lemma_M1_sobolev} Let $f\in L^{1}({\mathbb{R}^{d}})$ be such that $%
\partial _{x_{k}}f\in {M^1}({\mathbb{R}^{d}})$, for $k=1,\ldots ,d$. Then $%
f\in {M^1}({\mathbb{R}^{d}})$ and $\lVert f\rVert _{{M^1}}\lesssim \lVert
f\rVert _{L^{1}}+\sum_{k=1}^{d}\lVert \partial _{x_{i}}f\rVert _{{M^1}}$.
\end{lemma}

We can now derive the main technical result. Its statement uses the space of
$\mathrm{BV}(\Rtdst)$ of (integrable) functions of bounded variation; see
Section \ref{sec_bv} for some background.

\begin{theorem}
\label{th_change} Let $g_{1},g_{2}\in \mathcal{S}({\mathbb{R}^d})$ with $%
\lVert g_i\rVert_{2}=1$ and $m\in \mathrm{BV}({\mathbb{R}^{2d}})$. Then
\begin{equation*}
\lVert H_m^{g_1} - H_m^{g_2}\rVert_{S^{1}}\leq C_{g_{1},g_{2}}\mathit{var}%
(m),
\end{equation*}
where $C_{g_{1},g_{2}}$ is a constant that only depends on $g_1$ and $g_2$.
In particular, when $m=1_\Omega$ we obtain that
\begin{equation*}
\lVert H_\Omega^{g_1} - H_\Omega^{g_2}\rVert_{S^{1}}\leq C_{g_{1},g_{2}}
\left| \partial \Omega \right|  _{2d-1}.
\end{equation*}
\end{theorem}

\begin{proof}[Proof of Theorem \protect\ref{th_change}]
Let us assume first that $m$ is smooth and compactly supported. We use the
description of time-frequency localization operators as Weyl operators. By %
\eqref{eq_tf_weyl}, $H_{m}^{g_{i}}=(m\ast W(g_{i},g_i))^{w}$. Now, let $%
h:=W(g_{1},g_{1})-W(g_{2},g_{2})$. Then $h\in \mathcal{S}$ - see, e.g., \cite%
[Proposition 1.92]{folland89} - and $\int h=\lVert g_{1}\rVert
_{2}^{2}-\lVert g_{2}\rVert _{2}^{2}=0$ by~\eqref{eq_norm_wigner}. Hence, by
Proposition \ref{prop_s1},
\begin{equation*}
\lVert H_{m}^{g_{1}}-H_{m}^{g_{2}}\rVert _{S^{1}}=\lVert (m\ast h)^{w}\rVert
_{S^{1}}\lesssim \lVert m\ast h\rVert _{M^1},
\end{equation*}
Therefore, it suffices to prove that $\lVert m\ast h\rVert _{M^1}\lesssim
\mathit{var}(m)$. We apply Lemma \ref{lemma_M1_sobolev} to this end. First
note that $\partial _{x_{i}}(m\ast h)=\partial _{x_{i}}m\ast h$ and,
consequently,
\begin{equation*}
\lVert \partial _{x_{i}}(m\ast h)\rVert _{M^1}\lesssim \lVert \partial
_{x_{i}}m\rVert _{L^{1}}\lVert h\rVert _{M^1}\lesssim \mathit{var}(m).
\end{equation*}
Second, we exploit the fact that $\int h=0$ to get
\begin{align*}
(m\ast h)(z)& =\int_{\mathbb{R}^d}m(z^{\prime })h(z-z^{\prime })dz^{\prime
}=\int_{{\mathbb{R}^d}}(m(z^{\prime })-m(z))h(z-z^{\prime })dz^{\prime } \\
& =\int_{{\mathbb{R}^d}}\int_{0}^{1}\left\langle \nabla (m)(tz^{\prime
}+(1-t)z),z^{\prime }-z\right\rangle dt\,h(z-z^{\prime })dz^{\prime },
\end{align*}
and consequently
\begin{align*}
\int_{{\mathbb{R}^d}}\left\vert m\ast h(z)\right\vert dz& \leq
\int_{0}^{1}\int_{{\mathbb{R}^d}}\int_{{\mathbb{R}^d}}\left\vert \nabla
(m)(tz^{\prime }+(1-t)z)\right\vert \left\vert z^{\prime }-z\right\vert
\left\vert h(z-z^{\prime })\right\vert dz^{\prime }dzdt \\
& =\int_{0}^{1}\int_{{\mathbb{R}^d}}\int_{{\mathbb{R}^d}}\left\vert
\nabla(m)(tw+z)\right\vert \left\vert w\right\vert \left\vert
h(-w)\right\vert dwdzdt \\
& =\lVert \nabla m\rVert_{L^{1}}\int_{0}^{1}\int_{{\mathbb{R}^d}}\left\vert
w\right\vert \left\vert h(w)\right\vert dwdt= \lVert \nabla
m\rVert_{L^{1}}\int_{{\mathbb{R}^d}}\left\vert w\right\vert \left\vert
h(w)\right\vert dw.
\end{align*}
Since $h\in \mathcal{S}$ the last integral is finite. We conclude that $%
\lVert m*h\rVert_{L^{1}}\lesssim \lVert \nabla m\rVert_{L^{1}}=\mathit{var}%
(m)$, providing the argument for smooth, compactly supported $m$. For
general $m \in \mathrm{BV}({\mathbb{R}^d})$, there exists a sequence of
smooth, compactly supported functions $\left \{ m_k: k \geq 0 \right \} $
such that $m_k \rightarrow m$ in $L^1$, and $\mathit{var}(m_k) \rightarrow
\mathit{var}(m)$, as $k \rightarrow +\infty$ (see for example \cite[%
Sec.~5.2.2, Theorem 2]{evga92}.) By Proposition \ref{prop_s1}, $%
H^{g_i}_{m_k} \rightarrow H^{g_i}_m$ in trace norm, and the conclusion
follows by a continuity argument.
\end{proof}

\subsection{Comparison of correlation kernels}

We now state and prove the main result on the comparison between finite
WH ensembles associated with different subsets of eigenfunctions.

\begin{theorem}
\label{th_main_a} Consider the identification of the $(r,N)$-pure
polyanalytic ensemble as a finite WH ensemble with parameters $(h_r,
D_{R_N}, I_{r,N})$ given by Proposition \ref{prop_ident_2}. Let $K_{h_r,
D_{R_N},I_{r,N}}$ be the corresponding correlation kernel, and let $K_{h_r,
D_{R_N}}$ be the correlation kernel of the finite Weyl-Heisenberg ensemble
associated with the Hermite window $h_r$ and the disk $D_{R_N}$. Then
\begin{align}  \label{eq_tne_2}
\bignorm{K_{h_r, D_{R_N}} - K_{h_r, D_{R_N},I_{r,N}}}_{S^1} \lesssim \left|
\partial D_{R_N} \right|  _{1} \asymp \sqrt{N},
\end{align}
where $\norm{\cdot}_{S^1}$ denotes the trace-norm of the corresponding
integral operators.
\end{theorem}

\begin{proof}
\noindent \emph{Step 1: Comparison of different polyanalytic levels.} We
consider two eigen-expansions of the Toeplitz operator $M^{h_r}_{D_{R_N}}$:
\begin{align}  \label{eq_exp2}
M^{h_r}_{D_{R_N}} &= \sum_{j \geq 1} \lambda_j(D_{R_N}, h_r) \,
p^{D_{R_N}}_{h_r, j} \otimes p^{D_{R_N}}_{h_r, j}, \\
M^{h_r}_{D_{R_N}} &=\sum_{j \geq 0} \mu^r_{j, R_N} \, V_{h_r} h_j \otimes
V_{h_r} h_j. \label{eq_exp1}
\end{align}
Recall that, while the eigenvalues in \eqref{eq_exp1} are ordered
non-increasingly, the eigenvalues in \eqref{eq_exp2} follow the indexation
of Hermite functions. When $r=0$, according to Corollary \ref{coro_ident_3},
the two expansions coincide: the sequence $\mu^0_{j,R_N}$ is decreasing, and
\begin{align}  \label{eq_exp3}
\lambda_{j+1}(D_{R_N}, h_0)=\mu^0_{j,R_N}, \qquad j \geq 0.
\end{align}
We now quantify the deviation between the two eigen-expansions for general $r
$. To this end, we use the unitary equivalence between $M^{h_r}_{D_{R_N}}$
and the time-frequency localization operator $H^{h_r}_{D_{R_N}}$ - cf. %
\eqref{eq_diagram}. By \eqref{eq_hr},
\begin{equation*}
H_{D_{R_{N}}}^{h_{r}}=\sum_{j\geq 0}\mu_{j,R_N}^{r} \, h_j \otimes h_j.
\end{equation*}
While the operators $M^{h_r}_{D_{R_N}}$ act on mutually orthogonal subspaces
of $L^2(\Rtdst)$ for different values of $r$, their counterparts $%
H_{D_{R_{N}}}^{h_{r}}$ act on configuration space and so can readily be
compared by means of Theorem \ref{th_change}. We obtain
\begin{equation}
\lVert \mu_{\cdot, R_N} ^{0}-\mu_{\cdot, R_N} ^{r}\rVert _{\ell ^{1}}=\lVert
H_{D_{R_{N}}}^{h_{0}}-H_{D_{R_{N}}}^{h_{r}}\rVert _{S^{1}}\leq C_{r} \left|
\partial D_{R_{N}} \right|  _{1} \asymp R_N \asymp \sqrt{N}.  \label{eq_mus}
\end{equation}

\noindent \emph{Step 2. Estimates for the spectral truncations.} According
to Proposition \ref{prop_ident_2},
\begin{align}  \label{eq_xas}
K_{h_{r},D_{R_{N}},I_{r,N}} = \sum_{j=0}^{N-1} V_{h_r} h_j \otimes V_{h_r}
h_j.
\end{align}
For clarity, in what follows we denote by $T_K$ the operator with integral
kernel $K$. Let $L_j := 1$ for $1 \leq j \leq N$ and $L_j := 0$, for $j>N$.
Using the expansion in \eqref{eq_exp1} and \eqref{eq_diag_toep}, we estimate
the trace-norm:
\begin{align*}
&\norm{T_{K_{h_r, D_{R_N}}} - M^{h_r}_{D_{R_N}}}_{S^1} = \bignorm{ \sum_{j
\geq 1} \big(L_j - \lambda_j(D_{R_N}, h_r)\big) \, p^{D_{R_N}}_{h_r, j}
\otimes p^{D_{R_N}}_{h_r, j} }_{S^1} \\
&\qquad \leq \sum_{j \geq 1} \left| L_j - \lambda_j(D_{R_N}, h_r) \right|
= \sum_{j=1}^N \left[ 1-\lambda_j(D_{R_N}, h_r) \right] + \sum_{j>N}
\lambda_j(D_{R_N}, h_r) \\
&\qquad = N - \sum_{j \geq 1} \lambda_j(D_{R_N}, h_r) + 2 \sum_{j>N}
\lambda_j(D_{R_N}, h_r) = 2 \sum_{j>N} \lambda_j(D_{R_N}, h_r) \, ,
\end{align*}
as $\sum _j \lambda _j = |D_{{R_N}}| = N$ by \eqref{eq:o10}. Since $%
\mu^r_{j,R_N}$ is a rearrangement of $\lambda_j(D_{R_N}, h_r)$, we can use %
\eqref{eq_exp2} and \eqref{eq_xas} to mimic the argument. Thus, a similar
calculation gives
\begin{align*}
\norm{T_{K_{h_r, D_{R_N}, I_{r,N}}} - M^{h_r}_{D_{R_N}}}_{S^1} \leq 2
\sum_{j>N-1} \mu^r_{j,R_N},
\end{align*}
and consequently,
\begin{align}  \label{eq_almost}
\norm{T_{K_{h_r, D_{R_N}}} - T_{K_{h_r, D_{R_N}, I_{r,N}}}}_{S^1} \lesssim
\sum_{j>N} \lambda_j(D_{R_N}, h_r) + \sum_{j>N-1} \mu^r_{j,R_N}.
\end{align}
\noindent \emph{Step 3. Final estimates.} Combining \eqref{eq_almost} with %
\eqref{eq_exp3} and \eqref{eq_mus} we obtain
\begin{align}  \label{eq_almost_2}
\norm{T_{K_{h_r, D_{R_N}}} - T_{K_{h_r, D_{R_N}, I_{r,N}}}}_{S^1} \lesssim
\sum_{j>N} \lambda_j(D_{R_N}, h_r) + \sum_{j>N} \lambda_j(D_{R_N}, h_0) +  
\sqrt{N}.
\end{align}
We now invoke Lemma \ref{lemma_one_point_I} and Theorem \ref{th_quant_one}
to estimate
\begin{align}  \label{eq_almost_3}
\sum_{j>N} \lambda_j(D_{R_N}, h_r) \asymp \norm{\rho_{h_r, D_{R_N}} -
1_{D_{R_N}}}_{L^1} \lesssim \left| \partial D_{R_{N}} \right|  _{1} \asymp
\sqrt{N}.
\end{align}
Finally, \eqref{eq_tne_2} follows by combining \eqref{eq_almost_2} and %
\eqref{eq_almost_3}.
\end{proof}

\subsection{Transference to finite pure polyanalytic ensembles}

\begin{proof}[Proof of Theorem \protect\ref{th_main}]
We use Proposition \ref{prop_ident_2} to identify the $(r,N)$-polyanalytic
ensemble with a Weyl-Heisenberg ensemble with parameters $%
(h_{r},D_{R_{N}},I_{r,N})$, with correlation $K_{h_{r},D_{R_{N}},I_{r,N}}$
as in Theorem \ref{th_main_a}. By Proposition \ref{prop_ident_2}, $%
\widetilde K_{h_{r},D_{R_{N}},I_{r,N}} = K_{r,N}$. Therefore, the conclusion
follows from \eqref{eq_tne_2}.
\end{proof}

\subsection{The one-point intensity of finite polyanalytic ensembles}

\label{sec_xxx}

\begin{proof}[Proof of Theorem \protect\ref{asy_pure}]
We use the notation of Theorem \ref{th_main_a}; in particular $R_N = \sqrt{%
\tfrac{N}{\pi}}$. By \eqref{eq_related}, $\rho_{r,N}=%
\rho_{h_{r},D_{R_{N}},I_{r,N}}$, and we can estimate
\begin{align*}
\lVert \rho_{r,N}-1_{D_{R_N}}\rVert _{1} \leq \norm{%
\rho_{h_{r},D_{R_{N}},I_{r,N}} - \rho_{h_r,D_{R_N}}}_1 + \norm{%
\rho_{h_r,D_{R_N}} - 1_{D_{R_N}}}_1.
\end{align*}
By Theorem \ref{th_quant_one}, $\norm{\rho_{h_r,D_{R_N}} - 1_{D_{R_N}}}_1
\lesssim \sqrt{N}$. In addition, by Lemma \ref{lemma_trace} in the appendix,
\begin{align*}
\norm{\rho_{h_{r},D_{R_{N}},I_{r,N}} - \rho_{h_r,D_{R_N}}}_1 &= \int_{\Rtdst%
} \left| K_{h_{r},D_{R_{N}},I_{r,N}}(z,z) - K_{h_r,D_{R_N}}(z,z) \right|   dz
\\
&\leq \bignorm {K_{h_{r},D_{R_{N}},I_{r,N}} - K_{h_r,D_{R_N}}}_{S^1}.
\end{align*}
Hence, the conclusion follows from Theorem \ref{th_main_a}.
\end{proof}

Note that the proofs of Theorems~\ref{th_main_a} and \ref{asy_pure} combine
our main insights: the identification of the finite polyanalytic ensembles
with certain WH ensembles, the analysis of the spectrum of time-frequency
localization operators and Toeplitz operators, and the non-asymptotic
estimates of the accumulated spectrum.

\section{Double orthogonality}

\label{sec_do}

\subsection{Restriction versus localization}

\label{sec_duality} Let $\mathcal{X}^g$ be an infinite WH ensemble on $\Rtdst
$ and $\Omega \subseteq \Rtdst$ of finite measure and non-empty interior. We
consider the \emph{restriction operator} $T^g_\Omega: L^2(\Rtdst) \to L^2(%
\Rtdst)$,
\begin{align*}
T^g_\Omega F := 1_\Omega P_{\mathcal{V}_g} (1_\Omega \cdot F),
\end{align*}
and the \emph{inflated Toeplitz operator} $S^g_\Omega: L^2(\Rtdst) \to L^2(%
\Rtdst)$,
\begin{align*}
S^g_\Omega F := P_{\mathcal{V}_g} (1_\Omega \cdot P_{\mathcal{V}_g} F).
\end{align*}
In view of the decomposition $L^2(\Rtdst) = \mathcal{V}_{g} \oplus \mathcal{V%
}_{g}^\perp$, $S^g_\Omega$ and $M_{\Omega}^{g}$ are related by
\begin{align*}
S^g_{\Omega} = \left[
\begin{array}{cc}
M_{\Omega}^{g} & 0 \\
0 & 0%
\end{array}
\right],
\end{align*}
and therefore share the same non-zero eigenvalues, and the corresponding
eigenspaces coincide. The integral representation of $S^g_\Omega $ is given
by \eqref{Toe}. Since $P_{\mathcal{V}_g}$ and $F \mapsto F \cdot 1_\Omega$
are  orthogonal projections, both $T^g_\Omega$ and $S^g_\Omega$ are 
self-adjoint operators with spectrum contained in $[0,1]$. The integral
kernel of $T^g_\Omega$ is given by \eqref{eq_intro_ker_o} and $\int {K^g}%
_{|\Omega}(z,z) dz = \left| \Omega \right|   <+\infty$. Therefore, $%
T^g_\Omega$ is trace-class (see e.g. \cite[Theorems 2.12 and 2.14]%
{simonsbook}). It is an elementary fact that $T^g_\Omega$ and $S^g_\Omega$
have the same non-zero eigenvalues with the same multiplicities (this is
true for $PQP$ and $QPQ$ whenever $P$ and $Q$ are orthogonal projections).
Morever, for $\lambda \not=0$, the map
\begin{equation*}
F \longmapsto \frac{1}{\sqrt{\lambda}} 1_\Omega F
\end{equation*}
is an isometry between the eigenspaces
\begin{equation*}
\left \{ F \in L^2(\Rtdst): S^g_\Omega F = \lambda F \right \}
\longrightarrow  \left \{ F \in L^2(\Rtdst): T^g_\Omega F = \lambda F \right
\}  \, .
\end{equation*}
Therefore, if $M^g_\Omega$ is diagonalized as in \eqref{eq_eigenexp}, then $%
T^g_\Omega$ can be expanded as in \eqref{eq_intro_ker_o2}. This justifies
the discussion in Section \ref{sec_intro_cons}.

\subsection{Simultaneous observability}

\label{sec_so} Let $\mathcal{X}$ be a determinantal point process\ (with a
Hermitian locally trace-class correlation kernel). We say that a family of
sets $\left\{ \Omega_\gamma:\gamma \in \Gamma\right\}$ is \emph{%
simultaneously observable} for $\mathcal{X}$, if the following happens. Let $%
\Omega=\bigcup_{\gamma \in \Gamma }\Omega_\gamma$. There is an orthogonal
basis $\{\varphi_{j}:j\in J\}$ of the closure of the range of the
restriction operator $T_\Omega$ consisting of eigenfunctions of $T_{\Omega}$
such that for each $\gamma \in \Gamma$, the set $\{\varphi_{j}|_{\Omega_%
\gamma}:j\in J\}$ of the restricted functions is orthogonal. This is a
slightly relaxed version of the notion in \cite[pg. 69]{DetPointRand}: in
the situation of the definition, the functions $\{\varphi
_{j}|_{\Omega_\gamma}:j\in J\} \setminus \left \{ 0 \right \} $ form an
orthogonal basis of the closure of the range of $T_{\Omega_\gamma}$, but we
avoid making claims about the kernel of $T_\Omega$. As explained in \cite[%
pg. 69]{DetPointRand}, the motivation for this terminology comes from
quantum mechanics, where two physical quantities can be measured
simultaneously if the corresponding operators commute (or, more concretely,
if they have a basis of common eigenfunctions).

\begin{theorem}
\label{th_gab_sym} Let $\mathcal{D} = \left \{ D_R: R \in {\mathbb{R}}^+
\right \} $ be the family of all disks of ${\mathbb{R}}^2$ centered at the
origin and $r\in \bN$. Then

\begin{itemize}
\item[(i)] $\mathcal{D}$ is simultaneously observable for the infinite
Weyl-Heisenberg ensemble with window $h_{r}$.

\item[(ii)] Let $D_{R_0}$ be a disk and $I\subseteq {\mathbb{N}}$. Then $%
\mathcal{D}$ is simultaneously observable for the Weyl-Heisenberg ensemble $%
\mathcal{X}^{h_r}_{D_{R_0}, I}$.
\end{itemize}
\end{theorem}

\begin{proof}
Let us prove (i). Since the definition of simultaneous observability
involves the orthogonal complement of the kernels of the restriction
operators $T^g_{D_R}$, $\overline{\mathrm{ran}} (T_{D_R}^g) = (\mathrm{ker}%
\, T_{D_R}^g)^\perp $, the discussion in Section \ref{sec_duality} implies
that it suffices to show that the Toeplitz operators $M_{D_R}^{h_r}$ have a
common basis of eigenfunctions. Since $V_{h_r}^{\ast
}M_{D_R}^{h_r}V_{h_r}=H_{D_R}^{h_r}$, and, by Proposition \ref{prop_hn_disks}%
, the Hermite basis diagonalizes $H_{D_R}^{h_r}$ for all $R>0$, the
conclusion follows.

Let us now prove (ii). The ensemble $\mathcal{X}^{h_r}_{D_{R_0}, I}$ is
constructed by selecting the eigenfunctions of the Toeplitz operator $%
M_{D_{R_0}}^{h_r}:{\mathcal{V}_{h_r}}\rightarrow {\mathcal{V}_{h_r}}$
corresponding to the indices in $I$:
\begin{equation*}
K^{h_r}_{D_{R_0},I}(z,z^{\prime })=\sum_{j\in I}p^{D_{R_0}}_{h_r,j}(z)%
\overline{p^{D_{R_0}}_{h_r,j}(z^{\prime })}.
\end{equation*}
Since, by part (i), the functions $p^\Omega_{g,j}$ are orthogonal when
restricted to disks, the conclusion follows.
\end{proof}

As a consequence, we obtain Theorem \ref{th_intro_sym}, which we restate for
convenience.

\begin{reptheorem}{th_intro_sym}
The family $\mathcal{D} = \left \{ D_R: r \in {\mathbb{R}}^+ \right \} $ of all disks of
${\mathbb{C}}$ centered at the
origin is
simultaneously observable for every finite and infinite pure-type
polyanalytic ensemble.
\end{reptheorem}

\begin{proof}
This follows immediately from Proposition \ref{prop_ident_2} and Theorem \ref%
{th_gab_sym}. (This slightly extends a result originally derived by Shirai
\cite{SHIRAI}.)
\end{proof}

\subsection{An extension of Kostlan's theorem}

Theorem \ref{eq_thm4} is a consequence of the following slightly more
general result.

\begin{theorem}
\label{th_poly_kostlan} Let $\mathcal{X}$ be the determinantal point process
associated with the $(r,J)$-pure polyanalytic ensemble, with $J\subseteq \Nst%
_{0}$ finite. Then the point process on $[0,+\infty )$ of absolute values $%
\left\vert \mathcal{X}\right\vert $ has the same distribution as the process
generated by $\{Y_{j}: j\in J\}$ where the $Y_{j}$'s are independent random
variables with density
\begin{equation*}
f_{Y_{j}}(x):=2 \frac{\pi^{j-r+1} r!}{j!} x^{2(j-r)+1}\left[ L_{r}^{j-r}(\pi
x^2)\right]^{2}e^{-\pi x^2}.
\end{equation*}
(Hence, $Y^2_j$ is distributed according to $f_{Y_j^2}(x) = \frac{\pi^{j-r+1} r!}{j!} x^{j-r}\left[ L_{r}^{j-r}(\pi
x)\right]^{2}e^{-\pi x}$.)
\end{theorem}

\begin{proof}
We want to show that the point processes $\left\vert \mathcal{X}\right\vert
:=\sum_{x\in \mathcal{X}}\delta _{\left\vert x\right\vert }$ on
$\mathbb{R}$  and $\mathcal{Y}
:=\sum_{j\in J}\delta _{Y_{j}}$ on $\mathbb{C}$  have the same distribution. Let $%
I_{k}=[r_{k},R_{k}]$, $k=1,\ldots N$, be a disjoint family of subintervals
of $[0,+\infty )$. Then
\begin{equation*}
\left( \mathcal{Y}(I_{1}),\ldots ,\mathcal{Y}(I_{N})\right) \overset{d}{=}
\sum_{j\in J}\zeta_{j},
\end{equation*}
where the $\zeta_{j}$ are independent, $\mathbb{P}(\zeta_{j}=e_{k})=%
\int_{r_{k}}^{R_{k}}f_{Y_{j}}(x)dx$, and $\mathbb{P}(\zeta_{j}=0)=\int_{\Rst %
\setminus \cup_{k}[r_{k},R_{k}]}f_{Y_{j}}(x)dx$. On the other hand, Theorem %
\ref{th_intro_sym} implies that the annuli $A_k:= \left \{ z \in \bC : r_k
\leq \left| z \right|   \leq R_k \right \} $ are simultaneously observable
for $\mathcal{X}$. Hence, by \cite[Proposition 4.5.9]{DetPointRand} - which
is still applicable for the slightly more general definition of simultaneous
observability in Section \ref{sec_so}, we have
\begin{equation*}
\left( \left\vert \mathcal{X}\right\vert (I_{1}),\ldots ,\left\vert \mathcal{%
\ X}\right\vert (I_{N})\right) =\left( \mathcal{X}(A_{1}),\ldots ,\mathcal{X}
(A_{N})\right) \overset{d}{=}\sum_{j\in J}\zeta _{j}^{\prime },
\end{equation*}
where the $\zeta_{j}^{\prime }$ are independent, $\mathbb{P}%
(\zeta_{j}^{\prime }=e_{k})= \int_{A_k}\left\vert H_{j,r}(z,\overline{z}%
)\right\vert^{2} e^{-\pi \left\vert z\right\vert^{2}} dz$, and $\mathbb{P}%
(\zeta_{j}^{\prime }=0) = \int_{\bC \setminus \cup_k A_k}\left\vert
H_{j,r}(z,\overline{z})\right\vert^{2} e^{-\pi \left\vert z\right\vert^{2}}
dz$. A direct calculation, together with the identity
\begin{equation*}
\frac{(-x)^{k}}{k!}L_{r}^{k-r}(x)=\frac{(-x)^{r}}{r!}L_{k}^{r-k}(x)
\end{equation*}
shows that $\left(\zeta_{j}: j \in J \right) \overset{d}{=}
\left(\zeta_{j}^{\prime}: j \in J \right)$ and the conclusion follows.
\end{proof}

\begin{rem}
\rm{Let $n(R)$ denote the number of points of a point process in the
disk of radius $R$ centered at the origin. An immediate consequence of
Theorem \ref{th_poly_kostlan} is the following formula for the probability
of finding such a disk void of points, when the points are distributed
according to the a polyanalytic Ginibre ensemble of the pure type:
\begin{equation*}
\mathbb{P}\left[ n(R)=0\right] =\prod_{j}P\left( Y_{j} \geq R \right)
\end{equation*}
This is known as the hole probability (see \cite[Section 7.2]{DetPointRand}
for applications in the case of the Ginibre ensemble).}
\end{rem}

\appendix

\section{Additional background material}

\label{sec_app}

\subsection{Determinantal point processes and intensities}
\label{sec_det}

We follow the presentation of \cite{Bor00, DetPointRand}.
Let $K: \Rst^d \times \Rst^d \to \bC$ be a locally trace-class Hermitian
kernel with spectrum contained in $[0,1]$, and consider the functions
\begin{align}
\label{eq_rho}
\rho_{n}(x_1, \ldots, x_n) := \det \left( K(x_j, x_k)
\right)_{j,k=1,\ldots,d}, \qquad x_1, \ldots, x_n \in \Rdst.
\end{align}
The Macchi-Soshnikov theorem implies that there exists a point process $\mathcal{X}$ on $\Rst^d$ such that for every family of disjoint measurable
sets $\Omega_1, \ldots \Omega_n \subseteq \Rst^d$,
\begin{align*}
\mathbb{E} \left[ \prod_{j=1}^n \mathcal{X}(\Omega_j) \right] =
\int_{\prod_j \Omega_j} \rho_{n}(x_1,\ldots,x_n) dx_{1} \ldots dx_n,
\end{align*}
where $\mathcal{X}(\Omega)$ denotes the number of points of $\mathcal{X}$ to
be found in $\Omega$. The functions $\rho_{n}$ are known as correlation
functions or intensities and $\mathcal{X}$ is called a determinantal point
process. The one-point intensity $\rho$ is simply the diagonal of the
correlation kernel
\begin{equation*}
\rho(x) = \rho_{1}(x)=K(x,x),
\end{equation*}
and allows one to compute the expected number of points to be found on a
domain $\Omega$:
\begin{equation*}
\mathbb{E}\left[ \mathcal{X}(\Omega)\right] =\int_{\Omega}\rho (x)dx.
\end{equation*}
The one-point intensity can also be used to evaluate expectations of linear
statistics:
\begin{equation*}
\tfrac{1}{n}\mathbb{E} \left[f(x_{1})+\ldots+f(x_{n})\right]=\mathbb{E}\left[
f(x_{1})\right] =\int_{\Rdst}f(x)\rho(x)dx.
\end{equation*}

A DPP can be represented by different kernels. If $m: \Rst^d \to \bC$ is unimodular (i.e., $\abs{m(z)}=1$), then the kernel
\begin{align*}
K_m(x,x') = \overline{m}(x) K(x,x') m(x'),
\end{align*}
produces the same intensities in \eqref{eq_rho} as $K$ does. (This is a so-called \emph{gauge transformation}). The integral operator with kernel $K_m$ is related to the one with kernel $K$ by
\begin{align*}
\overline{m}(x) T_K (m f) (x)
= \int_{\Rdst} \overline{m}(x) K(x,x') m(x') f(x') dx'
= T_{K_m} f(x).
\end{align*}
Similarly, a linear transformation of a  DPP corresponds to a linear change of variables in the kernel $K$.

\subsection{Functions of bounded variation}

\label{sec_bv} A real-valued function $f\in L^{1}({\mathbb{R}^{d}})$ is said
to have \emph{bounded variation}, $f\in \mathrm{BV}({\mathbb{R}^{d}})$, if
its distributional partial derivatives are finite Radon measures. The
variation of $f$ is defined as
\begin{equation*}
\mathit{var}(f):=\sup \left\{ \int_{\mathbb{R}^{d}}f(x)\, \mathrm{div}\,
\phi (x)dx:\phi \in C_{c}^{1}({\mathbb{R}^{d}},{\mathbb{R}^{d}}),\left\vert
\phi (x)\right\vert _{2}\leq 1\right\},
\end{equation*}
where $C_{c}^{1}({\mathbb{R}^{d}},{\mathbb{R}^{d}})$ denotes the class of
compactly supported $C^{1}$-vector fields and $\mathrm{div}$ is the
divergence operator. If $f$ is continuously differentiable, then $f\in
\mathrm{BV}({\mathbb{R}^{d}})$ simply means that $\partial _{x_{1}}f,\ldots $%
, $\partial_{x_{d}}f\in L^{1}({\mathbb{R}^{d}})$, and $\mathit{var}(f)=\int_{%
\mathbb{R}^{d}}\left\vert \nabla f(x)\right\vert _{2}dx=\lVert \nabla
f\rVert _{L^{1}}$. A set $\Omega \subseteq {\mathbb{R}^{d}}$ is said to have
\emph{finite perimeter} if its characteristic function $1_{\Omega}$ is of
bounded variation, and the perimeter of $\Omega$ is defined as $\left|
\partial \Omega \right|  _{d-1} :=\mathit{var}(1_{\Omega}) $. If $\Omega$
has a smooth boundary, then $\left| \partial \Omega \right|  _{d-1}$ is just
the $(d-1)$-Hausdorff measure of the topological boundary. See \cite[Chapter
5]{evga92} for an extensive discussion of $\mathrm{BV}$.

\subsection{Trace-class operators}

\begin{lemma}
\label{lemma_trace} Let $K: \Rdst \times \Rdst \to \bC$ be a continuous
function and assume that the integral operator
\begin{align*}
T_K f(x) = \int_\Rdst K(x,y) f(y) dy, \qquad f \in L^2(\Rdst),
\end{align*}
is well-defined, bounded, and trace-class. Then $\int_\Rdst \left| K(x,x)
\right|   dx \leq \norm{T_K}_{S^1}$, where $\norm{\cdot}_{S^1}$ denotes the
trace-norm.
\end{lemma}

\begin{proof}
Let $T_K = \sum_j \mu_j \varphi_j \otimes \psi_j$, with $\mu_j \geq 0$ and $%
\{\varphi_j: j \geq 1\}$, $\{\psi_j: j \geq 1\}$ orthonormal. Then $K(x,y) =
\sum_j \mu_j \varphi_j(x) \overline{\psi_j(y)}$ for almost every $(x,y)$,
and we can formally compute
\begin{align*}
\int_\Rdst \left| K(x,x) \right|   dx &\leq \sum_j \mu_j \int_\Rdst \left|
\varphi_j(x) \right|   \left| \psi_j(x) \right|   dx \\
&\leq \sum_j \mu_j \left(\int_\Rdst \left| \varphi_j(x) \right|  ^2 dx
\right)^{1/2} \left(\int_\Rdst \left| \psi_j(x) \right|  ^2 dx \right)^{1/2}
\\
&=\sum_j \mu_j = \norm{T_K}_{S^1}.
\end{align*}
An approximation argument using the continuity of $K$ is needed to justify
the computations with the restriction of $K$ to the diagonal - see \cite[%
Chapters 1,2,3]{simonsbook} for related arguments.
\end{proof}

\subsection{Properties of modulation spaces}

\label{sec_mod} Recall the definition of the modulation space $M^1$ in %
\eqref{eq_def_m1}. It is well-known that, instead of the Gaussian function $%
\phi$, any non-zero Schwartz function can be used to define ${M^{1}}$,
giving an equivalent norm ~\cite{fe81-2}, \cite[Chapter 9]{Charly}. Using
this fact, the following lemma follows easily.

\begin{lemma}
\label{lemma_M1} Let $f\in {L^{2}({\mathbb{R}^{d}})}$. Then:

\begin{itemize}
\item[(i)] $f\in {M^1}({\mathbb{R}^{d}})$ if and only if $\hat{f}\in {M^1}({%
\mathbb{R}^{d}})$, where $\hat{f}$ is the Fourier transform of $f$. In this
case: $\lVert f\rVert _{{M^1}}\asymp \lVert \hat{f}\rVert _{{M^1}}$.

\item[(ii)] If $f$ is supported on $D_1(0)=\{x:\left| x \right|  \leq 1\}$
and $\hat{f}\in L^{1}({\mathbb{R}^{d}})$, then $f\in {M^1}({\mathbb{R}^{d}})$
and $\lVert f\rVert _{{M^1}}\lesssim \lVert \hat{f}\rVert _{L^{1}}$.

\item[(iii)] If $f\in {M^1}({\mathbb{R}^{d}})$ and $m\in C^{\infty }({%
\mathbb{R}^{d}})$ has bounded derivatives of all orders, then $m\cdot f\in {%
M^1}({\mathbb{R}^{d}})$, and $\lVert m\cdot f\rVert _{{M^1}}\leq C_{m}\lVert
f\rVert _{{M^1}}$, where $C_{m}$ is a constant that depends on $m $.
\end{itemize}
\end{lemma}

We now prove the Sobolev embedding lemma that was used in Section \ref%
{sec_change}.

\begin{proof}[Proof of Lemma \protect\ref{lemma_M1_sobolev}]
Let $g$ be such that $\hat{g}=f$. By Lemma \ref{lemma_M1}, it suffices to
show that $g\in {M^{1}}(\mathbb{R})$ and satisfies a suitable norm estimate.
Let $\eta \in C^{\infty }(\mathbb{R})$ be such that $\eta (\xi )\equiv 0$
for $\left\vert \xi \right\vert \leq 1/2$ and $\eta (\xi )\equiv 1$ for $%
\left\vert \xi \right\vert >1$. We write $\eta (\xi )=\sum_{k=1}^{d}\xi
_{k}\eta _{k}(\xi )$, where $\eta _{k}\in C^{\infty }(\mathbb{R})$ has
bounded derivatives of all orders. We set $g_{1}:=\eta \cdot g$ and $%
g_{2}:=(1-\eta )\cdot g$. Then $g_{1}(\xi )=\sum_{k=1}^{d}\eta _{k}(\xi )\xi
_{k}g(\xi )$. Since $\xi _{k}g(\xi )=\tfrac{1}{2\pi i}\widehat{\partial
_{x_{k}}f}(\xi )$ is in $M^1$ by Lemma~\ref{lemma_M1}(i) and $\eta _{k}$ has
bounded derivatives of all orders, we conclude from Lemma \ref{lemma_M1}%
(iii) that $g_{1}\in {M^{1}}(\mathbb{R})$ and that
\begin{equation*}
\norm{g_1}_{M^1} \asymp \norm{\widehat{g}_1}_{M^1} \lesssim \sum_{k=1}^{d} %
\norm{\xi_k \widehat{g}}_{M^1} \asymp \sum_{k=1}^{d}\lVert \partial
_{x_{k}}f\rVert _{{M^{1}}}.
\end{equation*}
On the other hand, since $g$ has an integrable Fourier transform, so does $%
g_{2}=(1-\eta )\cdot g$ and $\lVert \widehat{g_{2}}\rVert _{L^{1}}\lesssim
\lVert f\rVert _{L^{1}}$. In addition, $g_{2}$ is supported on $D_{1}(0)$.
Therefore, by Lemma \ref{lemma_M1}, $g_{2}\in {M^{1}}$ and $\lVert
g_{2}\rVert _{{M^{1}}}\lesssim \lVert f\rVert _{L^{1}}$. Hence $%
g=g_{1}+g_{2}\in {M^{1}}$, and it satisfies the stated estimate.
\end{proof}

\subsection{Polyanalytic Bargmann-Fock spaces}

\label{app_poly}

A complex-valued function $F(z,\overline{z})$ defined on a subset of $%
\mathbb{C}$ is said to be \emph{polyanalytic of order }$q-1$, if it
satisfies the generalized Cauchy-Riemann equations
\begin{equation}
\left( \partial _{\overline{z}}\right) ^{q}F(z,\overline{z})=\frac{1}{2^{q}}
\left( \partial _{x}+i\partial _{\xi }\right) ^{q}F(x+i\xi ,x-i\xi )=0 \, .
\label{eq:c1}
\end{equation}
Equivalently, $F$ is a polyanalytic function of order $q-1$ if it can be
written as
\begin{equation}
F(z,\overline{z})=\sum_{k=0}^{q-1}\overline{z}^{k}\varphi _{k}(z),
\label{polypolynomial}
\end{equation}
where the coefficients $\{\varphi _{k}(z)\}_{k=0}^{q-1}$ are analytic
functions. The \emph{polyanalytic Fock space} $\mathbf{F}^{q}(\mathbb{C})$
consists of all the polyanalytic functions of order $q-1$ contained in the
Hilbert space  $L^{2}(\mathbb{C},e^{-\pi \left\vert z\right\vert ^{2}})$.
The reproducing kernel of the polyanalytic Fock space $\mathbf{F}^{q}(%
\mathbb{C}) $ is
\begin{equation*}
\mathbf{K}^{q}(z,z^{\prime })=L_{q}^{1}(\pi \left\vert z-z^{\prime
}\right\vert ^{2})e^{\pi z \overline{z^{\prime }}}.
\end{equation*}
Polyanalytic Bargmann-Fock spaces appear naturally in vector-valued
time-frequency analysis \cite{Abreusampling}, \cite{CharlyYurasuper} and
signal multiplexing \cite{balan00, balan99}. Within $\mathbf{F}^{q}(\mathbb{C})$ we
distinguish the \emph{\ polynomial subspace}
\begin{equation*}
Pol_{\pi ,q,N}=span\{z^{j}\overline{z}^{l}:0\leq j\leq N-1,0\leq l\leq q-1\},
\end{equation*}
with the Hilbert space structure of $L^{2}(\mathbb{C},e^{-\pi \left\vert
z\right\vert ^{2}})$. The polyanalytic Ginibre ensemble, introduced in \cite%
{HendHaimi}, is the DPP with correlation kernel corresponding to the
orthogonal projection onto $Pol_{\pi ,q,N}$ (weighted with the Gaussian
measure). In \cite[Proposition 2.1]{HendHaimi} it is shown that
\begin{equation*}
Pol_{\pi ,q,N}=span\{H_{j,r}(z,\overline{z}):0\leq j\leq N-1,0\leq r\leq
q-1\},
\end{equation*}
where $H_{j,r}$ are the complex Hermite polynomials \eqref{ComplexHermite}.
Thus, the reproducing kernel of $Pol_{\pi ,q,N}$ can be written as
\begin{equation}
\mathbf{K}_{\pi ,N}^{q}(z,z^{\prime
})=\sum_{r=0}^{q-1}\sum_{j=0}^{N-1}H_{j,r}(z,\overline{z}) \overline{%
H_{j,r}(z^{\prime },\overline{z^{\prime }})}.  \label{eq_ker_poly_fock_qn}
\end{equation}

\subsection{Pure polyanalytic-Fock spaces}

\label{sec_pure_pol}

The pure polyanalytic Fock spaces $\mathcal{F}^{r}(\mathbb{C})$ have been
introduced by Vasilevski in \cite{VasiFock}, under the name of true
polyanalytic spaces. They are spanned by the complex Hermite polynomials of
fixed order $r$\ and can be defined as the set of polyanalytic functions $F$
integrable in $L^{2}(\mathbb{C},e^{-\pi \left\vert z\right\vert ^{2}})$ and
such that, for some entire function $H$ \cite{Abreusampling},
\begin{equation*}
F(z)=\left( \frac{\pi ^{r}}{r!}\right) ^{\frac{1}{2}}e^{\pi \left\vert
z\right\vert ^{2}}\left( \partial _{z}\right) ^{r}\left[ e^{-\pi \left\vert
z\right\vert ^{2}}H(z)\right].
\end{equation*}
Vasilevski~\cite{VasiFock} obtained the following decomposition of the
polyanalytic Fock space $\mathbf{F}^{q}(\mathbb{C})$ into pure components
\begin{equation}
\mathbf{F}^{q}(\mathbb{C})=\mathcal{F}^{0}(\mathbb{C})\oplus ...\oplus
\mathcal{F}^{q-1}(\mathbb{C}).  \label{orthogonal}
\end{equation}
Pure polyanalytic spaces are important in signal analysis \cite%
{Abreusampling} and in connection to theoretical physics \cite{AoP,
HendHaimi}. Indeed, they parameterize the so-called \emph{Landau levels},
which are the eigenspaces of the Landau Hamiltonian and model the
distribution of electrons in high energy states (see e.g. \cite[Section 2]%
{SHIRAI}, \cite[Section 4.1]{APRT}).

The complex Hermite polynomials \eqref{ComplexHermite} provide a natural way
of defining a polynomial subspace of the true polyanalytic space:
\begin{equation*}
\mathcal{P}ol_{\pi ,r,N}=span\{H_{j,r}(z,\overline{z}):0\leq j\leq N-1\}.
\end{equation*}
Thus,
\begin{equation*}
Pol_{\pi ,q,N}=\mathcal{P}ol_{\pi ,0,N}\oplus ...\oplus \mathcal{P}ol_{\pi
,q-1,N}.
\end{equation*}
The reproducing kernel of $Pol_{\pi ,r,N}$ is therefore
\begin{equation*}
\mathcal{K}_{r,\pi ,N}(z,z^{\prime })=\sum_{j=0}^{N-1}H_{j,r}(z,\overline{z}%
) \overline{H_{j,r}(z^{\prime },\overline{z^{\prime }})},
\end{equation*}
and the corresponding determinantal point processes have been introduced in
\cite{HendHaimi}.

\end{document}